\documentclass[12pt]{article}

\usepackage[T1]{fontenc}
\usepackage{lmodern}
\usepackage[letterpaper,margin=1in]{geometry}
\usepackage{amsfonts}
\usepackage{amsmath}
\usepackage{amssymb}
\usepackage{amsthm}
\usepackage{natbib}
\usepackage{setspace}
\usepackage[colorlinks=true,linkcolor=blue,citecolor=blue,urlcolor=blue]{hyperref}

\newtheorem{theorem}{Theorem}
\newtheorem{claim}{Claim}
\newtheorem{corollary}{Corollary}
\newtheorem{definition}{Definition}
\newtheorem{lemma}{Lemma}
\newtheorem{remark}{Remark}
\DeclareMathOperator{\supp}{supp}

\begin{document}

\title{Implementation with Uncertain Evidence\thanks{%
We are especially grateful to Esteban Peralta and Yifei Sun for their
insightful comments that have greatly improved the paper. We also thank
Dongwoo Lee, Claudio Mezzetti, Satoru Takahashi, and participants at various
conferences and seminars for helpful suggestions. Financial support from the
Singapore Ministry of Education Academic Research Fund Tier 1
(R-122-000-328-115) is gratefully acknowledged. Soumen acknowledges
financial support from the National Natural Science Foundation of China
Grant No. W253196.}}
\author{Soumen Banerjee\thanks{%
China Economics and Management Academy, Central University of Finance and
Economics, soumen08@gmail.com} \and Yi-Chun Chen\thanks{%
Department of Economics and Risk Management Institute, National University
of Singapore, ecsycc@nus.edu.sg}}
\date{}
\maketitle

\begin{abstract}
We study a full implementation problem with a state unknown to the designer
but known to agents, where agents have uncertain evidence privately drawn
from state-dependent distributions. This setting allows ``perfect
deceptions,'' where agents' reports can mimic the
evidence distribution of a false state which cannot also be directly refuted
by evidence, making differentiation impossible for any mechanism. This
yields our main result: a necessary and sufficient condition, No Perfect
Deceptions (NPD), for implementation in (mixed-strategy) Bayesian Nash
equilibria. The solution uses ``competing''
scoring rules to elicit ex ante evidence distributions while canceling out
in equilibrium, a reflection device that converts low-probability
refutations into direct incentives for full evidence submission without
penalties that scale with inverse refutation probabilities, and an
endogenous ``test allocation'' using bets
which separate states based on the evidence probabilities. The constructed
mechanism is finite and uses off-equilibrium transfers with maximum
magnitude of the order of the outcome. We also allow limited but nontrivial
uncertainty in the state dimension. A generalized NPD condition
characterizes exact implementation under nonexclusive information and yields
exact outcomes with vanishing expected transfers under information smallness.
\end{abstract}

\pagebreak

\section{Introduction}

Implementation with hard evidence studies when a designer can obtain a
desired outcome even though agents have private information and conflicting
interests, but can only send messages that are consistent with the evidence
available to them. Existing full-implementation results in this literature
mostly treat evidence as certain: at the true state, the set of evidence
available to each agent is known to the other agents and to the designer (%
\cite{KT2012,BL2012,BCS2020}). This assumption is convenient, but it is also
restrictive. In many environments, evidence is generated by prospecting,
investigation, or experimentation, and is therefore uncertain ex ante.
Agents may know the payoff-relevant state and yet be unsure about what
evidence other agents will be able to present. The aim of this paper is to
characterize implementation in such environments.

We study a benchmark environment in which the payoff-relevant state is
unknown to the designer but common knowledge among the agents. The social
choice function (SCF)\ specifies the allocation that the designer aims to
achieve and depends only on that payoff-relevant state. At each state,
agents privately and independently receive evidence drawn from a
distribution, so the \emph{availability} of evidence is uncertain. Evidence
is not itself payoff relevant; rather, it matters because it constrains what
agents can credibly claim. We ask when a social choice function can be fully
implemented in Bayesian Nash equilibrium, regardless of the agents' utility
functions. Because of this latter robustness property, our study focuses on
a novel approach to achieving full implementation: relying on the
distributional uncertainty of evidence rather than the agents' preference
variation across states.

A typical implementation mechanism in such settings involves a message space
consisting of various reports (some of which involve evidence), an outcome
function that chooses an allocation from the set of feasible outcomes as a
function of the agents' messages, and (optionally) some transfers which
incentivize the agents to make only those reports in equilibrium that yield
the socially desired outcome. Allowing uncertainty in evidence changes the
implementation problem in a fundamental way. With deterministic evidence, a
false state $s^{\prime }$ can be refuted at state $s$ whenever some agent
can present an article of evidence which does not arise at $s^{\prime }$,
and the key implementability restriction is measurability, which requires
some evidence variation between states in which the designer wishes to
implement different outcomes \citep{BL2012,BCS2020}. With stochastic
evidence, by contrast, it is not enough that the distribution of evidence
vary between two states. Agents may be able to mimic not only the
(cheap-talk) reports that are feasible in another state, but also the
distribution over evidence that others expect to observe there by
strategically \emph{withholding} articles of evidence. This creates a new
obstruction to implementation, which we call a \emph{perfect deception}. A
perfect deception from state $s$ to state $s^{\prime }$ is, roughly, a
profile of strategies at $s$ that reproduces the equilibrium behavior
expected at $s^{\prime }$. If such a deception exists while $s^{\prime }$
cannot be refuted by the evidence realized at $s$, we show that no mechanism
can distinguish the two states on the equilibrium path in a way that is
robust to the preferences of the agents.

Our first result shows that this obstruction is exactly the relevant one in
the baseline model where the state is common knowledge. We identify a
condition, \emph{No Perfect Deceptions} (NPD), and prove in Theorem~\ref%
{MAINR} that a social choice function is implementable in mixed-strategy
Bayesian Nash equilibrium if and only if it satisfies NPD. In particular,
whenever a false state can be perfectly mimicked from a true state without
being refuted by evidence which may arise (at the true state), the social
choice function must assign the same outcome to the two states. NPD is
strictly stronger than the natural stochastic analogue of the measurability
condition in \cite{BL2012} and characterizes which SCFs can be implemented
in an environment with uncertain evidence when the designer cannot rely on
exact knowledge of the preferences of agents. Our sufficiency proof allows
agents to have state-dependent preferences but proceeds on the assumption
that the designer knows only a uniform bound on the utilities of the agents
rather than the exact preference profile.

The sufficiency proof requires three new mechanism-design ideas. The first
is a use of \emph{competing} scoring rules to ensure that a common state is
reported by the agents with probability one. In the constructed mechanism,
each agent $i$ makes three state reports --- the first two are relevant here
--- the first indexing her own evidence distribution, and the second
indexing the evidence distribution of the agent to her right (say agent $j$%
). If agent $j$ reports her full realized evidence, a strictly proper
scoring rule makes agent $i$ choose a second state report that induces the
true evidence distribution for $j$. Agent $j$'s report is then cross-checked
against that report at the level of the induced evidence distribution. This
cross-check from $i$'s second report to $j$'s first report induces $j$ to
present the true distribution in her first report. This forms the central
cross-check in our mechanism ($j$'s evidence to $i$'s second report to $j$'s
first report). An offsetting (\emph{competing}) score makes the two scores
cancel realization by realization on the equilibrium path, so the scoring
comparison provides the required incentives while yielding no transfers in
equilibrium.

The second idea is a reflection device. With stochastic evidence, a lie may
be refutable only with small probability. Deterring the lie only by
punishing an actual refutation could require a penalty inversely
proportional to that probability.\footnote{%
If the bound on utilities is normalized to \$1, and the probability of being
refuted is 10\%, then such deterrence may require a penalty of $\geq \$10$.}
Instead, we ask each agent to forecast whether the agent to her right will
refute the current consensus. A positive forecast is used, in turn, to
incentivize agents to submit maximal evidence, which starts the scoring rule
cross-check described above, which realigns the profile to be consistent with
the truthful evidence distribution. Since the argument achieves consistency
in the state reports before handling refutation, the refutation probability
faced in equilibrium is one of finitely many values determined by the
primitives, and the forecasts may be restricted to this finite set (in the
baseline): a forecast then reveals the refutation probability.\footnote{%
Otherwise, the forecasts would need to track endogenous, arbitrarily fine
mixing probabilities, requiring an infinite mechanism.} This allows us to
use the same outcome-order cross-check penalty to handle refutation as well,
without requiring a penalty which scales up as the inverse of the minimum
refutation probability.

The third idea addresses nonrefutable lies. If the reports are consistent
with a false state that cannot be directly refuted, NPD implies that some
agent cannot reproduce the evidence distribution associated with that state.
A challenger can then file a challenge against that agent's evidence report,
using the third state report.\footnote{%
In our setup, this constitutes the analogue of the ``test allocation'' a la
Maskin.} The challenge (which is designed
to exploit the difference between the induced distribution of evidence and
the true distribution of evidence using the separating hyperplane theorem)
has negative expectation under truthful reporting at the claimed state but
positive expectation under every feasible attempt to mimic that state from
the truth. An active challenge incentivizes full evidence submission, and
again, the cross-check described above then realigns the profile to the truth.

We then relax the common-knowledge assumption on the state, allowing agents
to have uncertainty about one another's state types. Under nonexclusive
information (\cite{PS1986,V1999}), the other agents' state types determine
the focal agent's state type. Under information smallness (\cite%
{MP2002,MP2004}), they do so with high probability. This setting yields the 
\emph{Generalized No Perfect Deceptions} (GNPD) condition on local belief
closures. The extension uses higher-order reflection reports to propagate
alarms through belief closures and finitely many higher-order beliefs about
evidence to identify state types.\footnote{%
Here, the reflection reports are continuous, i.e., $[0,1]$, yielding an
infinite mechanism, which nevertheless continues to have well-defined best
responses and is thus distinct from integer or modulo games.} GNPD
characterizes exact implementation under nonexclusive information and
remains sufficient for exact outcomes with vanishing expected transfers
under information smallness.

Our baseline implementing mechanism works with two or more agents, does not rely on
integer or modulo games of the kind critiqued in \cite{J1992}, and uses
transfers that are zero in equilibrium and bounded for any bounded utility
domain. The baseline mechanism admits a natural pure-strategy equilibrium
based on truthful state reports and maximal evidence submission. Its message
space is finite; the mechanism for uncertain state types
has compact message spaces. The novel design
combines state reports
that index evidence distributions, scoring rules, cross-checks, reflection
reports, and challenge transfers. Although the theoretical construct is not
intended to be a simple, practical procedure, its role is to establish the
implementability frontier under the uncertain availability of hard evidence
without resorting to integer or modulo games.

The rest of the paper is organized as follows. Section~\ref{Example}
presents an illustrative example and clarifies the distinction between
signals and evidence. Section~\ref{Model} introduces the model, equilibrium
concept, and implementation notion. Section~\ref{BayesianImplementationR}
defines NPD and proves Theorem~\ref{MAINR}. Section~\ref{NCKS} studies
uncertain state types and presents the extension based on GNPD.

\section{\label{Example}Illustrative Example: Litigation}

This section illustrates the baseline environment studied in the paper. In
the baseline model, the agents know the state, but the designer does not,
and the agents are uncertain about one another's evidence. The social choice
function depends only on the state and not on the realized evidence.

Consider a damages case with two litigants, a plaintiff $P$ and a defendant $%
D$. The court must choose a monetary award, but does not observe the true
level of damages. The damage level is $s\in S=\{s^{L},s^{M},s^{H}\},$ where
the superscripts denote low, medium, and high damages, respectively. The
court's target social choice function is 
\begin{equation*}
f(s^{L})=a^{L},\qquad f(s^{M})=a^{M},\qquad f(s^{H})=a^{H},
\end{equation*}%
with $a^{L}<a^{M}<a^{H}$. Thus the socially desired outcome depends only on
the state $s$. Realized evidence may help the court implement $f$, but it is
not itself an argument of the social choice function.

In the baseline version of the example, both litigants know the true damage
level $s$, while the court does not. The litigants' preferences over awards
are independent of the state: the plaintiff always prefers a higher award,
and the defendant always prefers a lower award. Thus the state matters for
the court, but not directly for the litigants' ranking of outcomes. This is
precisely the kind of environment in which the designer cannot rely on
preference variation (only a constant social choice function is Maskin
monotonic) across states and must instead rely on evidence.

Only the plaintiff may possess hard evidence. The defendant has no hard
evidence, so his evidence endowment is degenerate. Formally, let
$E^{\emptyset}=\emptyset$ and $\mathbb{E}_{D}=\{E^{\emptyset}\}=\{\emptyset\}$,
and suppose that $p_{D}(s)(E^{\emptyset})=1$ for all $s\in S$. Hence the defendant's
message is pure cheap talk in the evidence dimension, and all of the
evidentiary structure lies on the plaintiff's side.

The plaintiff may possess verifiable articles of evidence of different
strengths. In the baseline version of the example, we use two such articles $%
e^{0}$ and $e^{MH}$. Their interpretation is as follows:

\begin{itemize}
\item $e^{0}$ is inconclusive and is compatible with all three damage levels 
$s^{L}$, $s^{M}$, and $s^{H}$;

\item $e^{MH}$ rules out low damages and is therefore compatible only with $%
s^{M}$ and $s^{H}$.
\end{itemize}

The plaintiff's \emph{evidence endowment} is the set of articles available
to her. We write $E^{0}=\{e^{0}\}$, and $E^{MH}=\{e^{MH},e^{0}\}$. If the
plaintiff is endowed with $E^{MH}$, she can always report $E^{0}$ instead by
withholding $e^{MH}$. This distinction between an evidence endowment and the
particular evidence report eventually made is central in our setting.

To isolate the effect of stochastic evidence, suppose that the plaintiff's
evidence endowment is drawn according to 
\begin{align*}
p_{P}(s^{L})(E^{0})& =1, \\
p_{P}(s^{M})(E^{MH})& =0.6,\qquad p_{P}(s^{M})(E^{0})=0.4, \\
p_{P}(s^{H})(E^{MH})& =0.4,\qquad p_{P}(s^{H})(E^{0})=0.6.
\end{align*}
This evidence structure already illustrates two basic notions used
throughout the paper. First, state $s^{L}$ is refutable at both $s^{M}$ and $%
s^{H}$, since the article $e^{MH}$ may appear there and rules out low
damages. Second, state $s^{H}$ is \emph{nonrefutable} at $s^{M}$. Indeed,
under state $s^{M}$ the plaintiff may be endowed with $E^{0}$ or $E^{MH}$,
and both are compatible with state $s^{H}$. Thus no realized article
available at $s^{M}$ directly proves that the state is $s^{M}$ rather than $%
s^{H}$.

The important point is that the lie ``the state is $s^{H}$'' is
nonrefutable at $s^{M}$, and it is also possible to
match the evidence distribution of $s^{H}$ exactly. This yields a perfect
deception from $s^{M}$ to $s^{H}$.

To see this, consider the direct mechanism in which each agent reports a
state and an evidence endowment. Since the defendant has no hard evidence,
his part of the deception is trivial: because there are no evidentiary
constraints on his message, he can always mimic whatever report distribution
he would use at state $s^{H}$. The nontrivial part concerns the plaintiff's
evidence report. Suppose the true state is $s^{M}$. If the plaintiff is
endowed with $E^{MH}$, let her report $(s^{H},E^{MH})$ with probability $2/3$
and report $(s^{H},E^{0})$ with probability $1/3$. If she is endowed with $%
E^{0}$, let her always report $(s^{H},E^{0})$. The induced distribution of
the plaintiff's reported evidence is then 
\begin{equation*}
\Pr(E^{MH})=0.6\times \frac{2}{3}=0.4,\qquad \Pr(E^{0})=0.6\times \frac{1}{3}%
+0.4=0.6,
\end{equation*}
which is exactly the truthful evidence distribution at state $s^{H}$.

Thus, under state $s^{M}$, the joint distribution of reports can be made
identical to the joint distribution that would arise at state $s^{H}$ if the
plaintiff were to report truthfully at state $s^{H}$.\footnote{%
It follows directly that the plaintiff can then mimic any (mixed) report she
would make to an arbitrary mechanism at state $s^{H}$.} The lie
``the state is $s^{H}$'' is nonrefutable at 
$s^{M}$, and the plaintiff can also reproduce the evidence distribution of $%
s^{H}$ exactly. This is a perfect deception. The example therefore shows why
a stochastic analogue of measurability is not sufficient. The evidence
distributions at $s^{M}$ and $s^{H}$ are different, yet that difference
alone does not make the two states separable by a mechanism.

A small perturbation destroys the perfect deception. Suppose now that at
state $s^{H}$ the plaintiff may also receive a more informative endowment 
\begin{equation*}
E^{H}=\{e^{H},e^{MH},e^{0}\},
\end{equation*}
where the article $e^{H}$ is compatible only with state $s^{H}$. Suppose
further that the evidence distribution at $s^{H}$ were to be 
\begin{equation*}
p_{P}(s^{H})(E^{H})=\delta ,\qquad p_{P}(s^{H})(E^{MH})=0.4-\delta ,\qquad
p_{P}(s^{H})(E^{0})=0.6,
\end{equation*}
for some small $\delta >0$, while the distributions at $s^{L}$ and $s^{M}$
remain unchanged.

State $s^{H}$ is still nonrefutable at $s^{M}$: the only endowments
available at $s^{M}$ are $E^{0}$ and $E^{MH}$, and neither refutes $s^{H}$.
But the perfect deception from $s^{M}$ to $s^{H}$ is now impossible. The
reason is feasibility. Under truthful reporting at $s^{H}$, the endowment $%
E^{H}$ is reported with positive probability, whereas at state $s^{M}$ the
plaintiff is never endowed with $E^{H}$ or any superset of it. Hence a
plaintiff at $s^{M}$ cannot mimic the truthful evidence distribution of $%
s^{H}$.

This perturbation captures the setting of the main theorem. To separate two
states with different desired outcomes, it is not necessary that the false
state be directly refutable. It is enough that a perfect deception be
impossible. In the first version of the example, state $s^{H}$ is
nonrefutable at $s^{M}$ and a perfect deception exists, so the two states
cannot be separated. In the perturbed version, state $s^{H}$ remains
nonrefutable at $s^{M}$, but the perfect deception fails because the
endowment $E^{H}$ cannot be mimicked.

A simple analogy helps explain how the mechanism exploits such a failure of
perfect deception. Suppose one agent draws a card from a deck and reports
the card, while another agent knows that the deck is missing one black card.
The second agent need not inspect the whole deck to make this information
credible. She can instead take the following bet: if the card is red, she
receives \$0.99; if the card is black, she pays \$1. If the deck were a full
deck, the expected value of the bet would be 
\begin{equation*}
\frac{26}{52}(0.99)+\frac{26}{52}(-1)<0.
\end{equation*}
But if one black card is missing, the expected value becomes 
\begin{equation*}
\frac{26}{51}(0.99)+\frac{25}{51}(-1)>0.
\end{equation*}
Thus the bet loses money under the claimed distribution but wins money under
the distribution that the challenger believes is feasible. In our setting, a
nonrefutable lie with no perfect deception has the same structure: no
realized article of evidence need directly refute the lie, but the
distribution of evidence reports that can be generated under the truth
differs from the distribution required by the lie. The challenge transfers
in the mechanism are bets of exactly this kind.

The mechanism developed later exploits exactly this distinction. Refutable
lies can be challenged by presenting evidence that contradicts them. By
contrast, nonrefutable lies can only be challenged through the distribution
of evidence. This is why the mechanism uses belief elicitation and challenge
transfers against reported evidence distributions: when a false state cannot
be ruled out by an article of evidence, it may still be impossible to
reproduce the evidence distribution associated with that state.

\section{\label{Model}Model}

This section presents the baseline environment used for Theorem \ref{MAINR}.
In the baseline, the payoff-relevant state is common knowledge among the
agents but unknown to the designer. The extension in which agents are also
uncertain about the state is deferred to Section \ref{NCKS}.

\subsection{Environment}

There is a finite set of agents $\mathcal{I}=\{1,\dots ,I\}$ with $I\geq 2$,
a finite set of states $S$, and a nonempty set of outcomes $A$. A \emph{%
social choice function} (SCF) is a map $f:S\rightarrow A$. Thus, the
socially desired outcome depends only on the payoff-relevant state and not
on the realized evidence.

Agents have quasilinear preferences. For each agent $i$, utility is given by 
\begin{equation*}
u_{i}(a,\tau _{i},s)=\bar{u}_{i}(a,s)+\tau _{i},
\end{equation*}%
where $a\in A$, $s\in S$, and $\tau _{i}\in \mathbb{R}$ is the monetary
transfer. A positive transfer is a payment \emph{to} the agent; a negative
transfer is a payment \emph{by} the agent. The designer does not know the
agents' utility functions, but does know a common bound on utility
differences. After normalization, we assume throughout that for every agent $%
i$, 
\begin{equation*}
\sup_{a,a^{\prime }\in A,\;s\in S}\left\vert \bar{u}_{i}(a,s)-\bar{u}%
_{i}(a^{\prime },s)\right\vert <1.
\end{equation*}%
The mechanisms below are designed based on this bound, but not on the
realized utility profile.

Each agent may possess hard evidence. Formally, for each agent $i$, $\mathbb{%
E}_{i}$ is a finite set consisting of all possible \emph{evidence endowments}%
. An element $E_{i}\in \mathbb{E}_{i}$ is an evidence endowment, namely a
finite set of articles of evidence available to agent $i$. For instance, $%
E^{MH}=\{e^{MH},e^{0}\}$ is an endowment in our litigation example in
Section \ref{Example}. We assume that $\mathbb{E}_{i}$ is closed under
subsets: if $E_{i}\in \mathbb{E}_{i}$ and $F_{i}\subseteq E_{i}$, then $%
F_{i}\in \mathbb{E}_{i}$. This allows us to use the same object both for a
realized endowment and for an evidence report which may be a subset of the
endowment if an agent withholds evidence when reporting to a mechanism.

At each state $s$, agent $i$'s evidence endowment is drawn from a
distribution $p_{i}(s)\in \Delta (\mathbb{E}_{i})$. Conditional on the state,
evidence draws are independent across agents, so the joint distribution of
evidence at state $s$ is 
\begin{equation*}
p(s)=\bigotimes_{i\in \mathcal{I}}p_{i}(s)\in \Delta (\mathbb{E}),\qquad 
\mathbb{E}=\prod_{i\in \mathcal{I}}\mathbb{E}_{i}.
\end{equation*}

Conditional independence is natural when agents receive separate
inspections, tests, records, or samples concerning a common underlying
condition.\footnote{%
For example, suppose two laboratories compete for a research grant and the
grantor can verify whether each possesses a viable vaccine candidate. Both
laboratories are likely to possess one when the disease is easy to target
and unlikely to do so when it is difficult, correlating their evidence
endowments ex ante through the common state. Conditional on that state,
however, their separate research outcomes may be independent as they have
separate research processes.} Conditional independence allows an agent's
state report, without conditioning on her realized evidence, to index
another agent's evidence distribution, as required by the scoring and
reflection transfers.

For later use, write 
\begin{equation*}
\mathbb{E}_{-i}=\prod_{j\neq i}\mathbb{E}_{j},\qquad
p_{-i}(s)=\bigotimes_{j\neq i}p_{j}(s),\qquad P_{i}=\{p_{i}(s):s\in
S\},\qquad P=\prod_{i\in \mathcal{I}}P_{i}.
\end{equation*}%
The timing in the baseline environment is as follows.

\begin{enumerate}
\item Nature selects a state $s\in S$.

\item The agents observe $s$; the designer does not.

\item Conditional on $s$, each agent $i$ receives an evidence endowment $%
E_i\in \mathbb{E}_i$ according to $p_i(s)$.

\item Each agent $i$ observes only her own realized endowment $E_i$.

\item Agents simultaneously send messages to a static mechanism defined in
Section \ref{sec:mech}.
\end{enumerate}

Thus, an agent's payoff-relevant private information consists of her
realized evidence endowment $E_{i}$, while uncertainty concerns the evidence
that other agents may have received. Because evidence is realized before
messages are sent, mixed strategies describe deliberate randomization \emph{%
after} the agent observes her endowment.

We impose the usual \emph{normality} requirement for one-shot mechanisms
with evidence (for instance in \cite{BL2012}): if agent $i$ is endowed with $%
E_{i}$, then she can submit any evidence report $F_{i}\subseteq E_{i}$,
including the full endowment $E_{i}$ itself. This rules out time or cost
constraints on full submission. Unlike in the deterministic model of \cite%
{BCS2020}, the full endowment cannot be reduced without loss to a single
most informative article: distinct endowments may share that article but
occur with different state-dependent probabilities, and those differences
may rule out a perfect deception.\footnote{%
In \cite{BCS2020}, articles of evidence are identified with the set of
states in which they appear, allowing for the definition of a
``tightest'' evidence, which combines all
available evidence into the smallest set of states which are consistent with
the evidence. In this paper, we eschew this representation in light of the
aforementioned difficulty - while several different endowments may yield the
same tightest evidence, a mechanism may nevertheless be able to exploit the
variation in the composition of these endowments.}

\subsection{\label{sec:mech}Mechanisms and Bayesian Nash equilibrium}

A finite mechanism is a tuple $\mathcal{M}=(M,g,\tau )$, where $%
M=\prod_{i\in \mathcal{I}}M_{i}$ is a finite message space, $g:M\rightarrow
A$ is the outcome rule, and $\tau=(\tau_i)_{i\in\mathcal{I}}$ is the transfer
rule, with $\tau_i:M\rightarrow\mathbb{R}$. Feasible messages may depend on
the agent's realized
evidence endowment. For each $E_{i}\in \mathbb{E}_{i}$, let $\varnothing
\neq M_{i}(E_{i})\subseteq M_{i}$ denote the set of messages feasible for an
agent endowed with $E_{i}$. These feasibility correspondences are understood
as part of the mechanism but are suppressed from the notation $\mathcal{M}%
=(M,g,\tau )$. We also require feasible messages to respect evidence
inclusion: 
\begin{equation}
E_{i}^{\prime }\subseteq E_{i}\quad \Longrightarrow \quad
M_{i}(E_{i}^{\prime })\subseteq M_{i}(E_{i}).  \label{eq:message-nesting}
\end{equation}%
Thus, possessing additional evidence never removes a feasible message.
Conceptually, messages sent to a mechanism may have cheap talk and evidence
components. Cheap talk components are those for which availability of
evidence does not constrain the choice of agents, whereas evidence
components are those where it does.

Fix a state $s\in S$ and a bounded utility profile $\bar{u}$. The mechanism $%
\mathcal{M}$ induces a Bayesian game. A mixed strategy for agent $i$ is a
map $\sigma _{i}:S\times\mathbb{E}_i\rightarrow \Delta (M_{i})$ such that 
\begin{equation*}
\supp\bigl(\sigma _{i}(s,E_{i})\bigr)\subseteq M_{i}(E_{i})\qquad \forall
(s,E_i)\in S\times\mathbb{E}_i.
\end{equation*}
A pure strategy is a degenerate mixed strategy.

Given a strategy profile $\sigma $, define 
\begin{align*}
U_{i}(m_{i},\sigma _{-i}\mid s,E_{i})={}&\sum_{E_{-i}\in \mathbb{E}%
_{-i}}p_{-i}(s)(E_{-i})\sum_{m_{-i}\in M_{-i}(E_{-i})}\sigma
_{-i}(s,E_{-i})(m_{-i}) \\
&\quad\times \left[\bar{u}_{i}(g(m_{i},m_{-i}),s)+\tau _{i}(m_{i},m_{-i})%
\right] ,
\end{align*}%
where $M_{-i}(E_{-i})=\prod_{j\neq i}M_{j}(E_{j})$ and $\sigma
_{-i}(s,E_{-i})$ is the product of the opponents' mixed strategies.

\begin{definition}
A strategy profile $\sigma$ is a Bayesian Nash equilibrium (BNE) of the game
induced by $\mathcal{M}$ at state $s$ if, for every agent $i$ and every
evidence endowment $E_i\in \mathbb{E}_i$ such that $p_i(s)(E_i)>0$, 
\begin{equation*}
\supp\bigl(\sigma_i(s,E_i)\bigr) \subseteq \arg\max_{m_i\in M_i(E_i)}
U_i(m_i,\sigma_{-i}\mid s,E_i).
\end{equation*}
\end{definition}

\subsection{Implementation}

We now state the implementation notion used in the paper.

\begin{definition}
An SCF $f:S\rightarrow A$ is \emph{implementable in mixed-strategy Bayesian
Nash equilibrium by a finite mechanism} if there exists a finite mechanism $%
\mathcal{M}=(M,g,\tau )$ such that, for every bounded utility profile $\bar{u%
}$ and every state $s\in S$, the induced game has at least one pure-strategy
BNE and every BNE $\sigma $ satisfies
\begin{equation*}
g(m)=f(s)\quad\text{and}\quad \tau _{i}(m)=0\quad\forall i\in \mathcal{I}
\end{equation*}%
for every $E\in \mathbb{E}$ such that $p(s)(E)>0$ and every $m\in
\prod_{i\in \mathcal{I}}\supp\bigl(\sigma _{i}(s,E_{i})\bigr)$.
\end{definition}

Thus the same finite mechanism delivers the prescribed outcome and zero
transfers at every equilibrium message profile. The mechanism must work
uniformly over the entire bounded utility domain described above. The
mechanism may depend on the commonly known bound on utilities but not on the
realized utility profile, a characteristic that we share with \cite{BL2012}.

\subsection{Refutation}

The key evidence notion is refutation. Because agents may withhold evidence,
a reported set of articles refutes a state only if that report cannot arise
at the state even as part of a larger endowment.

\begin{definition}
An evidence report $F_{i}\in \mathbb{E}_{i}$ \emph{refutes} state $s$
(written $F_{i}\downarrow s$) if no $E_i\in\mathbb{E}_i$ with $p_i(s)(E_i)>0$
contains $F_i$. If this condition fails, write $F_{i}\not\downarrow s$.
\end{definition}

A \emph{lie} at state $s$ is a false state claim $s^{\prime }\neq s$. A lie $%
s^{\prime }\in S$ is \emph{refutable by agent $i$ at state $s$} if there
exists $E_i\in\mathbb{E}_i$ with $p_i(s)(E_i)>0$ such that $E_i\downarrow
s^{\prime}$.

\begin{definition}
A lie $s^{\prime }\in S$ is \emph{refutable at state $s$} if it is refutable
by some agent $i\in \mathcal{I}$. Otherwise, it is \emph{nonrefutable at
state $s$}.
\end{definition}

If $s^{\prime }$ is nonrefutable at $s$, then every $E_i\in\mathbb{E}_i$
with $p_i(s)(E_i)>0$ is also compatible with $s^{\prime }$ in the sense of
the previous definition. In that case, direct refutation may be impossible
in either direction, and the implementation problem must instead be resolved
through the distribution of evidence.\footnote{%
Note that this is in contrast to the setup in \cite{BCS2020} wherein if $%
s^{\prime }$ is nonrefutable at $s$, then $s$ must be refutable at $%
s^{\prime }$.} With deterministic evidence, a nonrefutable lie between
states assigned different outcomes cannot also be a perfect deception: the
two properties would imply identical evidence profiles, contrary to
measurability. This is the source of the new obstruction with stochastic
evidence, which we study in the next section.

\section{\label{BayesianImplementationR} Mixed-Strategy Bayesian Nash
Implementation}

This section establishes the baseline characterization. The argument has two
parts. First, we identify \emph{perfect deceptions} -- the relevant
obstruction to implementation -- and prove necessity by restricting
attention to a convenient subclass of bounded preferences. Second, for any
SCF satisfying the NPD condition, we construct a mechanism that implements
it for the entire bounded utility domain.

\subsection{\label{NPDR}No Perfect Deceptions}

A natural first candidate for the implementing condition is the stochastic
analogue of the measurability condition in \cite{BL2012}: if $f(s)\neq
f(s^{\prime })$, then at least one agent's evidence distribution differs
between $s$ and $s^{\prime }$. Section \ref{Example} shows that this is too
weak. In the litigation example, state $s^{H}$ is nonrefutable at $s^{M}$,
yet the plaintiff can reproduce the evidence distribution of $s^{H}$ exactly
by selectively withholding evidence. Thus, what matters is not whether
evidence distributions differ, but whether the truthful distribution at one
state can be generated at another state by feasible withholding.

\begin{definition}
A \emph{deception} for agent $i$ is a map $\alpha _{i}:S\times\mathbb{E}_i
\rightarrow S\times \Delta (\mathbb{E}_{i})$ such that, at every $(s,E_i)\in
S\times\mathbb{E}_i$, its evidence-report distribution assigns positive
probability only to subsets of $E_i$.
\end{definition}

We identify each element of $S\times\Delta(\mathbb{E}_i)$ with the
corresponding lottery on $S\times\mathbb{E}_i$ that is degenerate in its
state coordinate, and use this identification whenever a deception is
evaluated as a lottery over state--evidence pairs.
Thus a deception specifies, for each $(s,E_i)$, a deterministic claimed
state and a randomized evidence report that can be generated from $E_i$ by
withholding. Its value at $(s,E_i)$ is immaterial when $p_i(s)(E_i)=0$. We
call each pair $(s^{\prime },F_i)$ in the support of $\alpha_i(s,E_i)$ a 
\emph{pseudo type}. A profile of deceptions is denoted by $\alpha =(\alpha
_{i})_{i\in \mathcal{I}}$.

\begin{definition}
Fix states $s,s^{\prime }\in S$. A deception $\alpha _{i}$ is \emph{perfect
from $s$ to $s^{\prime }$} if 
\begin{equation*}
\sum_{E_{i}\in \mathbb{E}_{i}}p_{i}(s)(E_{i})\alpha _{i}(s,E_{i})(s^{\prime
},E_{i}^{\prime })=p_{i}(s^{\prime })(E_{i}^{\prime })\qquad \forall
E_{i}^{\prime }\in \mathbb{E}_{i}.
\end{equation*}%
A deception profile $\alpha $ is perfect from $s$ to $s^{\prime }$ if this
holds for every agent $i$.
\end{definition}

Perfectness implies 
\begin{equation}
\supp\alpha_i(s,E_i)\subseteq \{s^{\prime }\}\times\supp p_i(s^{\prime })
\qquad\forall E_i\in\mathbb{E}_i\text{ such that }p_i(s)(E_i)>0.
\label{eq:perfect-supportR}
\end{equation}

In words, a perfect deception from $s$ to $s^{\prime }$ reproduces, agent by
agent, the same distribution of reports that truthful reporting would
generate at $s^{\prime }$. We write $s\rightarrow _{i}s^{\prime }$ if agent $%
i$ has a perfect deception from $s$ to $s^{\prime }$, and $s\not\rightarrow
_{i}s^{\prime }$ otherwise. We write $s\rightarrow s^{\prime }$ if every
agent has such a deception, and $s\not\rightarrow s^{\prime }$ otherwise.
Equivalently, $s\rightarrow _{i}s^{\prime } $ means that $p_{i}(s^{\prime })$
lies in the convex set of evidence-report distributions that agent $i$ can
generate at $s$ through feasible withholding. This reformulation will be
used below to construct challenge transfers.

Fix an agent $i$ and states $s,s^{\prime }\in S$. Call $\mathcal{U}\subseteq 
\mathbb{E}_{i}$ \emph{upward closed} if $F\in \mathcal{U}$ and $F\subseteq E$
imply $E\in \mathcal{U}$.

\begin{lemma}[Perfect deceptions, stochastic dominance, and challenges]
\label{lem:perfect-deception-characterization} The following are equivalent:

\begin{enumerate}
\item[(i)] $s\rightarrow_i s^{\prime}$.

\item[(ii)] The evidence distribution at $s$ dominates that at $s^{\prime}$
in the usual stochastic order on $(\mathbb{E}_i,\subseteq)$: 
\begin{equation}
p_i(s^{\prime})(\mathcal{U}) \leq p_i(s)(\mathcal{U}) \qquad \text{for every
upward-closed } \mathcal{U}\subseteq\mathbb{E}_i.
\label{eq:upper-set-dominance}
\end{equation}

\item[(iii)] There is no vector $b\in\mathbb{R}^{|\mathbb{E}_i|}$ such that $%
b\cdot p_i(s^{\prime})<0$ and $b\cdot q>0$ for every evidence-report
distribution $q$ obtainable from $p_i(s)$ by feasible withholding while
claiming $s^{\prime}$.
\end{enumerate}
\end{lemma}

The proof is in Appendix \ref{app:VERIFYNPD}. Condition (ii) gives a finite
test for perfect deception, while condition (iii) supplies the separating
challenge used below.

\begin{definition}
An SCF $f$ satisfies \textbf{No Perfect Deceptions (NPD)} if $%
f(s)=f(s^{\prime })$ whenever $s^{\prime }$ is nonrefutable at $s$ and there
is a perfect deception from $s$ to $s^{\prime }$.
\end{definition}

NPD is directional: It may be possible to mimic the evidence distribution of 
$s^{\prime }$ at $s$ but not the other way around, so the condition does not
partition the state space. This asymmetry is already visible in the example
from Section \ref{Example}, where there is a perfect deception from $s^{M}$
to $s^{H}$, but not vice versa.

The notion of a\ perfect deception is closest to that of a consistent
deception of \cite{Matsushima1993}.\footnote{%
In \cite{J1991}, Bayesian monotonicity is formulated in terms of deceptions:
profiles of pure mappings from true types to pseudo types, equivalently pure
reporting strategies in the direct mechanism. \cite{Matsushima1993} retains
this pure form but imposes consistency at every true type. Our deception
permits mixing over evidence reports conditional on realized evidence, while
the state claim remains pure.} A deception is consistent if, at every true
type, the distribution of opponents' reports it induces coincides with the
distribution that the type corresponding to the mimicked pseudo type expects
under truthful reporting. A perfect deception carries this
distribution-matching idea to stochastic evidence. Because the
payoff-relevant state is common knowledge in our baseline model, it requires
each agent to reproduce the state-dependent distribution of evidence
reports. Those reports must be feasible through evidence withholding, and
exact mimicry obstructs implementation only when the mimicked state is
nonrefutable.

\begin{theorem}
\label{MAINR} An SCF $f:S\rightarrow A$ is implementable in mixed-strategy
Bayesian Nash equilibrium by a finite mechanism if and only if it satisfies
NPD.
\end{theorem}

\subsection{\label{subsec:necc}Necessity}

Suppose that a mechanism $\mathcal{M}=(M,g,\tau )$ implements $f$. A bounded
utility profile is \emph{state independent} if, for each agent $i$, there
exists a function $v_{i}:A\rightarrow \mathbb{R}$ such that 
\begin{equation*}
\bar{u}_{i}(a,s)=v_{i}(a)\qquad \forall a\in A,\;s\in S.
\end{equation*}%
Because implementation must hold for every bounded utility profile, it must
in particular hold for every state-independent utility profile. Fix such a
profile and pick states $s,s^{\prime }\in S$ such that $s^{\prime }$ is
nonrefutable at $s$ and there exists a perfect deception $\alpha $ from $s$
to $s^{\prime }$.

Nonrefutability of $s^{\prime }$ at $s$ implies the following monotonicity
of feasible evidence reports: 
\begin{equation}
p_i(s)(E_i)>0\Longrightarrow \exists E_i^{\prime }\supseteq E_i\text{ such
that }p_i(s^{\prime })(E_i^{\prime })>0.  \label{eq:nonrefutable_claimR}
\end{equation}%
Indeed, if no such $E_{i}^{\prime }$ existed, then $E_{i}\downarrow
s^{\prime }$, so $s^{\prime }$ would be refutable at $s$.

Now let $\sigma $ be any BNE of the game induced by $\mathcal{M}$ at state $%
s^{\prime }$. Construct $\hat{\sigma}$ by leaving $\sigma$ unchanged away
from state $s$ and defining it at $s$ as follows. For every $E_i\in \mathbb{E%
}_i$ such that $p_i(s)(E_i)>0$, agent $i$ first draws a pseudo type $%
(s^{\prime },E_{i}^{\prime })$ according to $\alpha _{i}(s,E_{i})$, and then
plays exactly as $(s^{\prime },E_{i}^{\prime })$ plays under $\sigma$. By (%
\ref{eq:perfect-supportR}), every such $E_i^{\prime}$ occurs with positive
probability at $s^{\prime }$. Because the deception assigns positive
probability only to $E_{i}^{\prime }\subseteq E_{i}$, this copied strategy
is feasible by (\ref{eq:message-nesting}). At every $E_i\in\mathbb{E}_i$
such that $p_i(s)(E_i)=0$, choose any feasible strategy. Because the
deception is perfect, each agent faces under $\hat{\sigma}$ at state $s$ the
same distribution over opponents' messages as under $\sigma $ at $s^{\prime }
$.

Consider any agent $i$ and any $E_i\in\mathbb{E}_i$ such that $p_i(s)(E_i)>0$%
. Exact implementation at $s^{\prime }$ implies that every
equilibrium-support action used at an endowment that occurs with positive
probability at $s^{\prime }$ yields the outcome $f(s^{\prime })$ and zero
transfers against every opponent support profile. Hence every such action
gives continuation payoff $v_i(f(s^{\prime }))$. The same is true of every
action copied by $\hat{\sigma}$.

Suppose, toward a contradiction, that type $(s,E_i)$ has a profitable
deviation $d_i\in M_i(E_i)$. It must give strictly more than $%
v_i(f(s^{\prime })) $. Perfectness implies that the opponents' message
distribution under $\hat{\sigma}$ at $s$ is the same as under $\sigma$ at $%
s^{\prime }$, and state independence implies that $d_i$ has the same payoff
against these two distributions. By (\ref{eq:nonrefutable_claimR}), there
exists $E_i^{\prime \prime}\in\mathbb{E}_i$ with $p_i(s^{\prime
})(E_i^{\prime\prime})>0$ and $E_i^{\prime\prime}\supseteq E_i$. The nesting
condition makes $d_i$ feasible for $E_i^{\prime\prime}$. But type $%
(s^{\prime},E_i^{\prime\prime})$ obtains $v_i(f(s^{\prime}))$ in
equilibrium, while $d_i$ gives strictly more, contradicting the optimality
of $\sigma_i(s^{\prime},E_i^{\prime\prime})$.

Therefore $\hat{\sigma}$ is a BNE at state $s$. By construction, it induces
the same outcome as $\sigma $, namely $f(s^{\prime })$, with zero transfers
on its support. Since $\mathcal{M}$ implements $f$, we must have $%
f(s)=f(s^{\prime })$. Hence $f$ satisfies NPD.

\begin{remark}[Transfers and Necessity of NPD]
NPD remains necessary for implementation regardless of whether the mechanism
can use off-path transfers. The proof above only relies on the requirement
that any financial transfers must be zero \textit{on path}.
\end{remark}

\subsection{\label{CHALLENGER}Challenge Transfers}

We now detail how agents can issue challenges against coordinated
nonrefutable lies when a perfect deception does not exist. Suppose the
reports are consistent with a nonrefutable false state $s^{\prime }$. If the
true state is $s$ and no perfect deception exists from $s$ to $s^{\prime }$,
then some agent cannot reproduce the evidence distribution $p_{i}(s^{\prime
})$. The mechanism exploits this by letting an agent file a \emph{challenge}
against the current consensus. The designer then conditions a transfer on
the challenged agent's evidence report. That transfer is constructed to be
unprofitable under truthful reporting at $s^{\prime }$ but profitable under
every feasible attempt to mimic $s^{\prime }$ from state $s$.

Fix agent $i$ and states $s,s^{\prime }\in S$. Let $\mathcal{P}%
_{i}(s,s^{\prime })$ denote the set of evidence-report distributions that
agent $i$ can induce at state $s$ while claiming state $s^{\prime }$: \ 
\begin{equation*}
\mathcal{P}_{i}(s,s^{\prime })=\left\{
\begin{aligned}
q\in \Delta (\mathbb{E}_{i}):\quad &\exists \alpha _{i}\text{ feasible such that}\\
&q(E_{i}^{\prime })=\sum_{E_{i}\in \mathbb{E}_{i}}p_{i}(s)(E_{i})
\alpha _{i}(s,E_{i})(s^{\prime },E_{i}^{\prime }),\\
&\hspace{8em}\forall E_{i}^{\prime }\in \mathbb{E}_{i}
\end{aligned}
\right\} .
\end{equation*}%
By Lemma~\ref{lem:perfect-deception-characterization}, 
\begin{equation}
\begin{split}
s\not\rightarrow _{i}s^{\prime }\quad \Longleftrightarrow \quad & \text{%
there exists }b_{i}(s,s^{\prime })\in \mathbb{R}^{|\mathbb{E}_{i}|}\text{
such that} \\
& b_{i}(s,s^{\prime })\cdot p_{i}(s^{\prime })<0\quad \text{and}\quad
b_{i}(s,s^{\prime })\cdot q>0\quad \forall q\in \mathcal{P}_{i}(s,s^{\prime
}).
\end{split}
\label{eq:challenge-separation}
\end{equation}%
For each pair $(s,s^{\prime })$ with $s\not\rightarrow s^{\prime }$, choose
one agent 
\begin{equation*}
\hat{\imath}(s,s^{\prime })\in \left\{ i\in \mathcal{I}:s\not\rightarrow
_{i}s^{\prime }\right\}
\end{equation*}%
and fix one separating vector $b_{\hat{\imath}(s,s^{\prime })}(s,s^{\prime
}) $ satisfying \eqref{eq:challenge-separation}.

For each fixed $s$, choose these objects to depend on the claimed state $%
s^{\prime }$ only through $p(s^{\prime })$. This is possible because both
the set of eligible challenged agents and the corresponding separation
problems are identical whenever $p(s^{\prime })=p(t^{\prime })$. The
restriction makes the challenge transfer well-defined when a report profile
is consistent with more than one state. When the current consensus is $%
s^{\prime }$ and an agent claims that the true state is $s$, the mechanism
interprets this as a challenge directed at agent $\hat{\imath}(s,s^{\prime })
$'s evidence report. By \eqref{eq:challenge-separation}, we can construct a
challenge transfer (from the challenger to the designer) such that the
challenger loses money if the consensus state $s^{\prime }$ is in fact true
and the challenged agent reports truthfully (i.e., reveals her evidence
fully), but makes money under every feasible attempt to mimic $s^{\prime }$
from state $s$.

\subsection{\label{MECHANISMR}Implementing Mechanism}

We now define the implementing mechanism for the sufficiency proof of
Theorem \ref{MAINR}. Agents are seated around a circular table, so that for
each agent $i$, there is an agent to her right, namely agent $i+1$ (with $%
I+1=1$ ), and an agent to her left, namely agent $i-1$ (with $0=I$). Each
agent makes three state reports. The first state report indexes her own
reported evidence distribution, while the second indexes the evidence
distribution of the agent to her right. A strictly proper scoring rule
rewards the second report using the realized evidence submission of the
agent being predicted. A cross-check requires the predicted agent's first
report and the predicting agent's second report to induce the same evidence
distribution. It prevents mixing across different induced distributions.
Reports within the same distribution fiber are treated as equivalent, since
neither the outcome nor the subsequent proof steps require a particular
state label within a fiber. The third state report is used for challenges.

\subsubsection*{Message Space}

The message space for agent $i$ is $M_i=S\times S\times S\times C^R\times
\mathbb{E}_i$, where, for each $i\in\mathcal{I}$ and $s,s^{\prime}\in S$, we let
\begin{equation*}
r_{i}(s,s^{\prime })=\sum_{\{E_{i}:E_{i}\downarrow s^{\prime
}\}}p_{i}(s)(E_{i})
\end{equation*}%
denote the probability at state $s$ that agent $i$'s endowment refutes $%
s^{\prime }$, and let 
\begin{equation*}
C^{R}=\{r_{i}(s,s^{\prime }):i\in \mathcal{I},\ s,s^{\prime }\in S\}\cup
\{0\}.
\end{equation*}%
Since each factor of $M_{i}$ is finite, $M_{i}$ is a finite set.

A typical message is $m_i=(s_i^1,s_i^2,s_i^b,c_{i,i+1},E_i)$.
If $\bar{E}_{i}$ is agent $i$'s realized evidence endowment, a message $%
m_{i} $ is feasible if and only if its evidence component satisfies $%
E_{i}\subseteq \bar{E}_{i}$. The first report $s_{i}^{1}$ indexes agent $i$%
's own evidence distribution, and the second report $s_{i}^{2}$ indexes
agent $i+1$'s evidence distribution. The report $c_{i,i+1}$ forecasts
whether agent $i+1$ will refute a consistent state claim; $E_{i}$ is an
evidence report; and $s_{i}^{b}$ is used to file a challenge.

\subsubsection*{Outcome}

The outcome is $f(s)$ if there is a state $s$ such that
$p_i(s_i^1)=p_i(s)$ for every $i\in\mathcal{I}$, and is a fixed default
outcome $a^0\in A$ otherwise.

A message profile $m$ is \emph{consistent with state $s$} if, for every
agent $i$, both $p_i(s_i^1)$ and $p_i(s_{i-1}^2)$ equal $p_i(s)$.
Otherwise, it is \emph{inconsistent}. This outcome rule is well-defined when 
$f$ satisfies NPD. Indeed, if two states $s$ and $s^{\prime }$ induce the
same profile of evidence distributions, then they are mutually nonrefutable
and each perfectly deceives the other, so NPD implies $f(s)=f(s^{\prime })$.

\subsubsection*{Transfers}

The mechanism uses five transfers. The first transfer rewards evidence
submission in three situations: when the state reports are inconsistent,
when there is an active challenge, and when the agent's evidence refutes a
consistent state claim, where we say that an \emph{active challenge} is
present if the reports are consistent with some state $s^{\prime }$ and some
agent $i$ reports a challenge state $s_{i}^{b}$ such that $%
s_{i}^{b}\not\rightarrow s^{\prime }$ \emph{and} $\hat{\imath}%
(s_{i}^{b},s^{\prime })\neq i$. The report selects the separating vector for
the challenged agent $\hat{\imath}(s_{i}^{b},s^{\prime })$. It also contains
the reflection reward: if $c_{i-1,i}>0$, then agent $i$ has a positive
marginal incentive to submit additional evidence even when her realized
evidence does not itself refute the consensus: \ 

\begin{equation*}
\tau _{i}^{1}(m)=\left\{ 
\begin{array}{ll}
\varepsilon |E_{i}|, & 
\begin{array}{l}
\text{if }m\text{ involves inconsistency, or an active challenge,} \\ 
\text{or if }m\text{ is consistent with some }s\text{ and }E_{i}\downarrow s,%
\end{array}
\\[1.1em] 
\varepsilon c_{i-1,i}|E_{i}|, & \text{otherwise, if }c_{i-1,i}>0, \\ 
0, & \text{otherwise.}%
\end{array}%
\right.
\end{equation*}%
The second transfer is built from the quadratic scoring rule. For a reported
distribution $q\in \Delta (\mathbb{E}_{i})$ and a realized evidence report $%
E_{i}\in \mathbb{E}_{i}$, let $Q(q,E_{i})=2q(E_{i})-q\cdot q$, where $q\cdot
q=\sum_{E_{i}^{\prime }\in \mathbb{E}_{i}}q(E_{i}^{\prime })^{2}$; see \cite%
{brier1950verification}. This transfer uses the modified strictly proper
scoring rule: 
\begin{equation*}
\tau _{i}^{2}(m)=\gamma \left[
Q(p_{i+1}(s_{i}^{2}),E_{i+1})-Q(p_{i+1}(s_{i+1}^{1}),E_{i+1})\right] .
\end{equation*}%
This transfer is a competing scoring rule and gives agent $i$ the incentive
to choose a state report that induces agent $i+1$'s true evidence
distribution when agent $i+1$ reports all her evidence. On the equilibrium
path, $p_{i+1}(s_{i}^{2})=p_{i+1}(s_{i+1}^{1})$, so the two scores cancel
pointwise even if the underlying state labels differ.

The third transfer cross-checks each agent's first state report with the
second state report made by the agent to her left: 
\begin{equation*}
\tau _{i}^{3}(m)=\left\{ 
\begin{array}{ll}
-\kappa , & \text{if }p_{i}(s_{i}^{1})\neq p_{i}(s_{i-1}^{2}), \\ 
0, & \text{otherwise.}%
\end{array}%
\right.
\end{equation*}%
The fourth transfer is the reflection transfer. Define

\begin{equation*}
R_{i+1}=\left\{ 
\begin{array}{ll}
1, & \text{if }m\text{ is consistent with some state }s\text{ and }%
E_{i+1}\downarrow s \\ 
0, & \text{otherwise.}%
\end{array}%
\right.
\end{equation*}

Then, 
\begin{equation*}
\tau _{i}^{4}(m)=-2\varepsilon (c_{i,i+1}-R_{i+1})^{2}.
\end{equation*}%
The reflection transfer elicits agent $i$'s forecast of whether agent $i+1$
will refute the consensus.\footnote{%
If a message profile is consistent with more than one state, those states
have identical evidence supports. An evidence report therefore refutes
either all or none of them, so the refutation events used below are well
defined.} A positive forecast activates the reflection reward in $\tau
_{i+1}^{1}$, extending the full-submission incentive to all types of agent $%
i+1$.

The fifth transfer is the challenge transfer: 
\begin{equation*}
\tau _{i}^{5}(m)=\left\{ 
\begin{array}{ll}
\varepsilon \,b_{\hat{\imath}(s_{i}^{b},s^{\prime })}(s_{i}^{b},s^{\prime
})(E_{\hat{\imath}(s_{i}^{b},s^{\prime })}) & \text{if }m\text{ is
consistent with }s^{\prime }\text{ , }s_{i}^{b}\not\rightarrow s^{\prime }%
\text{, }i\not=\hat{\imath}(s_{i}^{b},s^{\prime }) \\ 
0, & \text{otherwise.}%
\end{array}%
\right.
\end{equation*}%
Agent $i$'s total transfer is $\tau_i(m)=\sum_{\ell=1}^{5}\tau_i^\ell(m)$.
That is, if agent $i$ challenges a report profile consistent with state $%
s^{\prime }$ using the challenge report $s_{i}^{b}$, then the designer
applies the appropriate separating challenge transfer against the evidence
report of $j=\hat{\imath}(s_{i}^{b},s^{\prime })$.

The reports $s_{i}^{b}$ and $c_{i,i+1}$ thus serve a common purpose. An
active $s_{i}^{b}$-challenge activates the full-submission incentive for
every agent, while $c_{i-1,i}>0$ activates it for every type of agent $i$,
including types whose own endowment does not refute the consensus. This
adapts the (cheap-talk based) contagion-of-truth logic in \cite{MAM} to
agents whose feasible evidence reports depend on their realized endowments.

The next subsection chooses the scale of the separating vectors and the
three constants in the transfer rule. The parameter $\varepsilon $ sets the
common scale of $\tau ^{1}$, $\tau ^{4}$, and $\tau ^{5}$, $\gamma $ weights
the competing score $\tau ^{2}$, and $\kappa $ is the cross-check penalty in 
$\tau ^{3}$.

\subsubsection*{Scaling}

The parameters must satisfy three comparisons. A false challenge must lose
more through $\tau ^{5}$ than it can gain through $\tau ^{1}$. Correcting a
second state report must gain more through $\tau ^{2}$ than it can lose
through $\tau ^{1}+\tau ^{4}+\tau ^{5}$. Correcting a first state report
must gain enough through $\tau ^{3}$ to cover the change in outcome utility
and in the other transfers. We now choose the parameters to satisfy these
comparisons.

We begin by choosing separating vectors satisfying %
\eqref{eq:challenge-separation}. Since each $b_{j}(s,s^{\prime })$ is
defined only up to positive scalar multiplication, we may choose it so that
for every pair $(s,s^{\prime })$ with $s\not\rightarrow s^{\prime }$, 
\begin{equation}
-b_{\hat{\imath}(s,s^{\prime })}(s,s^{\prime })\cdot p_{\hat{\imath}%
(s,s^{\prime })}(s^{\prime })>\max_{k\in \mathcal{I}}\max_{E_{k}\in \mathbb{E%
}_{k}}|E_{k}|.  \label{eq:bnormalR}
\end{equation}%
To see why this normalization is needed, consider a challenger $i$ and the
challenged agent $j=\hat{\imath}(s,s^{\prime })\neq i$. Conditional
independence implies that, at state $s^{\prime }$, the challenger evaluates
the challenged agent's evidence using $p_{j}(s^{\prime })$, even after
observing her own evidence. Filing the challenge can increase her evidence
reward by at most $\varepsilon |E_{i}|$. Inequality \eqref{eq:bnormalR}
makes her expected loss from $\tau _{i}^{5}$ larger than this possible gain.

Once the separating vectors are fixed, $\tau _{i}^{1}+\tau _{i}^{4}+\tau
_{i}^{5}$ is $\varepsilon $ times a function of $(i,m)$ that is independent
of $\varepsilon $. Since $\mathcal{I}\times M$ is finite, this function is
uniformly bounded: there exists $\zeta >0$ such that, for every $\varepsilon
>0$, 
\begin{equation}
|\tau _{i}^{1}(m)+\tau _{i}^{4}(m)+\tau _{i}^{5}(m)|\leq \zeta \varepsilon
\qquad \forall i\in \mathcal{I},\ m\in M.  \label{eq:small-transfer-boundR}
\end{equation}%
Fix any $\varepsilon >0$. Recall that each $P_{i}$ is the finite set of
state-indexed evidence distributions for agent $i$. If some $P_{i}$ contains
distinct distributions, let 
\begin{equation*}
\eta =\min \left\{ \lVert q^{\prime }-q\rVert _{2}^{2}:i\in \mathcal{I},\
q,q^{\prime }\in P_{i},\ q^{\prime }\neq q\right\} >0.
\end{equation*}%
Under the quadratic scoring rule, $\lVert q-q^{\prime }\rVert _{2}^{2}$ is
exactly the expected loss from reporting $q^{\prime }\neq q$ when the
evidence distribution is $q$.

Set $\gamma =3\zeta \varepsilon /\eta $, so that 
\begin{equation}
\gamma \eta =3\zeta \varepsilon >2\zeta \varepsilon .
\label{eq:score-scaleR}
\end{equation}%
If no $P_{i}$ contains distinct distributions, this comparison is not
needed, and we set $\gamma =\varepsilon $.

Finally, choose $\kappa $ so that 
\begin{equation}
\kappa >1+2\max_{i\in \mathcal{I},\;m\in M}\left\vert \tau _{i}^{1}(m)+ \tau
_{i}^{2}(m)+\tau _{i}^{4}(m)+\tau _{i}^{5}(m)\right\vert .
\label{eq:crosscheck-scaleR}
\end{equation}%
All four transfers inside the maximum are proportional to $\varepsilon $, so 
$\kappa $ can be chosen to converge to one as $\varepsilon $ converges to
zero. For example, choose $\kappa $ equal to the right-hand side of %
\eqref{eq:crosscheck-scaleR} plus $\varepsilon $.

Changing a first state report leaves the agent's own competing score
unchanged. The gain from removing a cross-check mismatch is therefore $%
\kappa $, while the resulting outcome-utility change is strictly less than
one. The change in the remaining transfers is at most twice the maximum in %
\eqref{eq:crosscheck-scaleR}. Thus removing a mismatch is strictly
profitable.

\subsection{Sufficiency: Proof of Implementation}

In what follows, the true state is $s^\ast$, and we write $%
p_i^\ast=p_i(s^\ast)$. We show that every equilibrium of the game induced by
the mechanism yields outcome $f(s^\ast)$ and zero transfers.

\subsubsection*{Preliminary Results}

\begin{claim}
\label{mainR} Fix an agent $i$ and an equilibrium. (i) Fix $E_i$ with
$p_i^\ast(E_i)>0$ and a message $m_i\in\supp\sigma_i(s^\ast,E_i)$.
Holding the other components of $m_i$ fixed, full submission of $E_i$ is
weakly optimal. It is strictly optimal if $\tau_i^1$ is positive at the full
report with positive probability over the other agents' equilibrium
messages. (ii) If every type
of agent $i$ submits her endowment in full with
probability one, then 
\begin{equation*}
p_{i}(s_{i}^{1})=p_{i}(s_{i-1}^{2})=p_{i}^{\ast }
\end{equation*}%
on every message profile in equilibrium support.
\end{claim}

\begin{proof}
(i) Fix $E_i$ and $m_i$ as in part (i), holding the other components of
$m_i$ fixed. Her evidence report enters her own
payoff only through $%
\tau _{i}^{1}$: the outcome and $\tau _{i}^{3}$ depend on state reports
alone, $\tau _{i}^{2}$ and $\tau _{i}^{4}$ are evaluated at $E_{i+1}$, and
an active challenge requires $\hat{\imath}(s_{i}^{b},s^{\prime })\neq i$, so 
$\tau _{i}^{5}$ is evaluated at another agent's evidence. Within $\tau
_{i}^{1}$, refutation is monotone. Full submission therefore weakly raises
both the applicable rate, which is $\varepsilon $, $\varepsilon
c_{i-1,i}\leq \varepsilon $, or zero, and the size of the evidence report.
It is strictly better than submitting any proper subset if it makes $\tau
_{i}^{1}$ positive with positive probability. This proves part (i).

(ii) Under the premise of part (ii), agent $i$'s evidence report has
distribution $p_{i}^{\ast }$. If a support report $s_{i-1}^{2}$ induces a
distribution different from $p_{i}^{\ast }$, agent $i-1$ can replace it with 
$s^{\ast }$. By \eqref{eq:score-scaleR}, the proper score gives a gain
greater than $2\zeta \varepsilon $. By \eqref{eq:small-transfer-boundR}, the
change in $\tau _{i-1}^{1}+\tau _{i-1}^{4}+\tau _{i-1}^{5}$ is at most $%
2\zeta \varepsilon $. The deviation changes neither the outcome nor agent $%
i-1$'s cross-check transfer. Hence every support report satisfies $%
p_{i}(s_{i-1}^{2})=p_{i}^{\ast }$.

If agent $i$ uses a first report with $p_i(s_i^1)\neq p_i^\ast$, she incurs
the cross-check loss $-\kappa$. Replacing it with $s^\ast$ removes that loss
and leaves her own competing score unchanged. By \eqref{eq:crosscheck-scaleR}%
, the gain dominates the resulting change in outcome utility and in the
remaining transfers. Therefore $p_i(s_i^1)=p_i(s_{i-1}^2)=p_i^\ast$.
\end{proof}

\subsubsection*{Consistency}

\begin{claim}
\label{CONSISTENCYR} In any equilibrium, there is a state $s\in S$ such that
every message profile in equilibrium support satisfies $%
p_i(s_i^1)=p_i(s_{i-1}^2)=p_i(s)$ for every $i$. Thus agents may mix only
among state reports that induce the same evidence distribution.
\end{claim}

\begin{proof}
We say that agent $i$'s induced-distribution report varies if $Z_{i}=\bigl(%
p_{i}(s_{i}^{1}),p_{i+1}(s_{i}^{2})\bigr)$ takes more than one value across
messages used with positive probability by types. Conditional independence
of evidence and independent mixing imply product support: any supported
value of each agent's $Z_{i}$ can occur jointly with positive probability.

Suppose inconsistency occurs with positive probability. If no $Z_{i}$
varies, inconsistency occurs with probability one. For each agent $i$,
every realized endowment, and every message in the support of that type's
equilibrium strategy, full submission therefore makes $\tau _{i}^{1}$
positive whenever the endowment is nonempty. Part (i) of Claim \ref{mainR}
gives full evidence submission with probability one, and part (ii) gives $%
p_{i}(s_{i}^{1})=p_{i}(s_{i-1}^{2})=p_{i}^{\ast }$ for every $i$,
contradicting inconsistency.

If at least two of the $Z_{i}$ vary, fix any agent $k$, any type of that
agent, and any message in the support of that type's equilibrium strategy.
Some agent $%
j\neq k$ has two distinct supported values of $Z_{j}$. Holding fixed one
supported value for every other agent, these two values cannot both satisfy
the adjacent distribution cross-checks. Product support therefore gives
agent $k$ positive probability of inconsistency when the other components
of that message are fixed.
Parts (i) and (ii) of Claim \ref{mainR} again yield a contradiction.

It remains to consider that exactly one $Z_{i}$, say that of agent $i$,
varies. Every type of every agent $j\neq i$, holding fixed the other
components of any message in the support of that type's equilibrium
strategy, assigns positive probability to
inconsistency. Part (i) of Claim \ref{mainR} therefore gives full evidence
submission with probability one for every $j\neq i$. Part (ii) then gives $%
p_{j}(s_{j}^{1})=p_{j}(s_{j-1}^{2})$ for every $j\neq i$ and, because agent $%
i+1$ submits all her evidence, also $p_{i+1}(s_{i}^{2})=p_{i+1}(s_{i+1}^{1})$%
. Hence only the distribution induced by $s_{i}^{1}$ can vary.

Let $q_{i}$ denote the fixed distribution $p_{i}(s_{i-1}^{2})$, and choose $%
t\in S$ with $p_{i}(t)=q_{i}$. Any support message satisfying $%
p_{i}(s_{i}^{1})\neq q_{i}$ incurs $\tau _{i}^{3}=-\kappa $. Replacing $%
s_{i}^{1}$ by $t$ removes this loss. By \eqref{eq:crosscheck-scaleR}, the
deviation is strictly profitable after accounting for the outcome and all
other transfers, a contradiction. Thus inconsistency occurs with probability
zero. It remains to show that the same evidence-distribution profile is used
throughout equilibrium support. Let $X_{i}=p_{i}(s_{i}^{1})$, and $%
Y_{i}=p_{i}(s_{i-1}^{2})$. The variables $X_{i}$ and $Y_{i}$ are controlled
by distinct agents. By product support, every supported value of $X_{i}$ can
be paired with every supported value of $Y_{i}$. Since any mismatch would
produce inconsistency, their supports must be the same singleton, say $%
\{q_{i}\}$. Hence every equilibrium-support profile induces the same vector $%
q=(q_{i})_{i}$. Since every such profile is consistent with some state,
there exists $s\in S$ such that $q_{i}=p_{i}(s)$ for every $i$.
\end{proof}

\subsubsection*{Eliminating Refutable Lies}

\begin{claim}
\label{NOORLR} There is no equilibrium in which the reports are consistent
with a refutable lie.
\end{claim}

\begin{proof}
Suppose the reports are consistent with a false state $s\neq s^{\ast }$ that
is refutable at $s^{\ast }$. There is an agent $i$ and an endowment $E_i\in%
\mathbb{E}_i$ with $p_i(s^\ast)(E_i)>0$ such that $E_{i}\downarrow s$.
Whenever agent $i $ receives a refuting endowment, the refutation branch of $%
\tau _{i}^{1}$ is positive at the full report when the other components of
any message in the support of her equilibrium strategy at that endowment
are fixed. Part (i) of Claim \ref{mainR} therefore makes full
submission strictly optimal for that type.

By Claim \ref{CONSISTENCYR}, the induced distribution reports are pure, so
the reports are consistent with $s$ with probability one. In any
equilibrium, agent $i$ thus submits every refuting endowment fully, while no
report from a nonrefuting endowment refutes $s$; her report therefore
refutes $s$ precisely when her endowment does. Conditional independence of
evidence implies that, after every support message of agent $i-1$, the
probability that agent $i$ refutes $s$ is exactly $r_{i}(s^{\ast },s)\in
C^{R}$, with $r_{i}(s^{\ast },s)>0$. Then, the unique optimal reflection
report is $c_{i-1,i}=r_{i}(s^{\ast },s)>0$. The reflection branch of $\tau
_{i}^{1}$ is consequently positive at the full report when the other
components of any message in the support of agent $i$'s equilibrium
strategy at her endowment are fixed. Part (i) of Claim \ref{mainR}
therefore makes every type of agent $i$ submit in full with probability one.

Part (ii) of Claim \ref{mainR} now gives $p_i(s_i^1)=p_i^\ast$. Consistency
with $s$ gives $p_i(s_i^1)=p_i(s)$. But refutability by agent $i$ implies $%
p_i^\ast\neq p_i(s)$, a contradiction.
\end{proof}

\subsubsection*{Eliminating Nonrefutable Lies}

\begin{claim}
\label{NONRLR} There is no equilibrium in which the reports are consistent
with a nonrefutable lie $s\neq s^\ast$ such that $f(s)\neq f(s^\ast)$.
\end{claim}

\begin{proof}
By Claim \ref{CONSISTENCYR}, every equilibrium report profile is consistent
with some state $s$. By Claim \ref{NOORLR}, if $s\neq s^\ast$ then $s$ must
be nonrefutable at $s^\ast$.

Suppose now that $f(s)\neq f(s^{\ast })$. Since $f$ satisfies NPD, the
selected agent $j=\hat{\imath}(s^{\ast },s)$ satisfies $s^{\ast
}\not\rightarrow _{j}s$. Let $q$ be the equilibrium distribution of agent $j$%
's evidence report, ignoring her other message components. Because her
evidence reports are feasible, she could use the same evidence-reporting
rule while claiming $s$. Hence $q\in\mathcal{P}_{j}(s^{\ast},s)$.

Pick an agent $i\neq j$. By conditional independence, every type of agent $i$
evaluates the challenge $s_{i}^{b}=s^{\ast }$ using the same distribution $q$%
. The challenge-separation condition therefore makes agent $i$'s expected
challenge transfer strictly positive. If agent $i$'s current challenge
report is inactive, changing only that report to $s^{\ast }$ gives this
positive expected transfer, can only increase $\tau _{i}^{1}$, and changes
no other part of her payoff. Thus every supported challenge report of agent $%
i$ is active.

An active challenge is therefore present with probability one. Part (i) of
Claim \ref{mainR} makes every agent submit her endowment in full, including
agent $j$. Part (ii) and consistency with $s$ then give 
\begin{equation*}
p_j(s^{\ast})=p_j(s_j^1)=p_j(s).
\end{equation*}%
Agent $j$ can consequently claim $s$ at $s^{\ast}$ and submit her full
endowment, exactly reproducing the evidence distribution at $s$. Hence $%
s^{\ast}\rightarrow_j s$, contradicting the choice $s^{\ast}\not%
\rightarrow_j s$.
\end{proof}

\subsubsection*{Implementation}

\begin{claim}
\label{IMPLR} In any equilibrium, every support message profile yields the
outcome $f(s^{\ast })$ and zero transfers for every agent.
\end{claim}

\begin{proof}
By Claims \ref{CONSISTENCYR}--\ref{NONRLR}, in any equilibrium the reports
are consistent with some state $s$ satisfying $f(s)=f(s^{\ast })$. Hence
every support message profile yields $g(m)=f(s^{\ast })$. It remains to show
that all transfers are zero on the equilibrium path.

Consistency gives $p_{i}(s_{i}^{1})=p_{i}(s_{i-1}^{2})$ for every agent, so $%
\tau _{i}^{3}(m)=0$. It also gives $p_{i+1}(s_{i}^{2})=p_{i+1}(s_{i+1}^{1})$%
, so the two scores in $\tau _{i}^{2}$ cancel pointwise even when the
underlying state labels differ, so that $\tau _{i}^{2}$ is zero.

There is no refutation on the equilibrium path. If the consistent state is
the truth, every submitted evidence report is compatible with it. If it is a
false state, Claim \ref{NOORLR} establishes nonrefutability. Thus $R_{i+1}=0$
at every support message profile. Reporting $c_{i,i+1}=0$ is then the unique
optimum in $C^R$, and $\tau_i^4(m)=0$ for every agent.

We now claim that there is no active challenge in equilibrium. Suppose
instead that some type of agent $i$ files an active challenge with positive
probability against agent $j=\hat{\imath}(s_{i}^{b},s^{\prime })$. By
definition, $j\neq i$. Whenever agent $i$ files the challenge, the first
branch of $\tau _{j}^{1}$ is positive. Conditional independence implies that
every type of agent $j$, holding fixed the other components of any message
in the support of that type's equilibrium strategy, assigns positive
probability to this event. Part (i) of
Claim \ref{mainR} therefore makes agent $j$ submit her evidence in full with
probability one. Part (ii), together with consistency with $s^{\prime }$,
gives 
\begin{equation*}
p_{j}(s^{\ast })=p_{j}(s_{j}^{1})=p_{j}(s^{\prime }).
\end{equation*}%
Conditional independence therefore implies that, after observing her own
evidence, agent $i$ evaluates the challenge bet using $p_{j}(s^{\prime })$.
By \eqref{eq:bnormalR}, 
\begin{equation*}
-\varepsilon b_{j}(s_{i}^{b},s^{\prime })\cdot p_{j}(s^{\prime
})>\varepsilon \max_{k,E_{k}}|E_{k}|\geq \varepsilon |E_{i}|.
\end{equation*}%
Withdrawing the challenge can reduce agent $i$'s evidence reward by at most $%
\varepsilon |E_{i}|$, changes no other part of her payoff, and avoids the
negative challenge transfer. Withdrawing is therefore strictly profitable, a
contradiction. Hence no active challenge occurs in equilibrium.

Finally, consistency together with the absence of challenges and
refutations, and the fact that every reflection report is zero, gives $%
\tau_i^1(m)=0$ for every agent. Therefore all equilibrium transfers vanish.
\end{proof}

It remains to establish the pure-strategy equilibrium existence required by
our implementation notion.

\begin{claim}
\label{PUREBNER}
If $f$ satisfies NPD, then, at every state $s^\ast$ and for every bounded
utility profile, the constructed mechanism has a pure BNE in which each agent
$i$ with endowment $E_i$ sends
\[
m_i=(s^\ast,s^\ast,s^\ast,0,E_i).
\]
\end{claim}

\begin{proof}
At this profile, the outcome is $f(s^\ast)$ and all transfers vanish. Fix a
type with $p_i^\ast(E_i)>0$ and any feasible deviation. If
$p_i(s_i^1)\neq p_i^\ast$, the cross-check loss dominates every possible
gain by \eqref{eq:crosscheck-scaleR}. If $p_i(s_i^1)=p_i^\ast$ but
$p_{i+1}(s_i^2)\neq p_{i+1}^\ast$, the expected scoring loss is at least
$\gamma\eta>2\zeta\varepsilon$, which dominates the remaining transfer
change. Otherwise, the outcome is unchanged, $\tau_i^2=\tau_i^3=0$, and
$\tau_i^4\leq 0$. Without an active challenge, $\tau_i^1=\tau_i^5=0$.
If agent $i$ files an active challenge against
$j=\hat{\imath}(s_i^b,s^\ast)\neq i$, conditional independence and
\eqref{eq:bnormalR} give
\[
\mathbb{E}[\tau_i^5]
<-\varepsilon\max_{k,E_k}|E_k|
\leq-\tau_i^1.
\]
Thus no deviation is profitable.
\end{proof}

This equilibrium has a natural interpretation: each agent reports the true
state in all three state reports, forecasts no refutation, and submits all
of her evidence. Truth-telling together with maximal evidence submission is
therefore an equilibrium of the mechanism, and all transfers are zero along
it.

Together with Claim \ref{IMPLR}, this establishes the sufficiency part of
Theorem \ref{MAINR}.

\subsection{Discussion}

\begin{corollary}[Bounded-domain implementation]
\label{COR_BOUNDED_DOMAINR} If an SCF $f$ satisfies NPD, then, for any
commonly known bound on utility differences, there exists a {finite}
mechanism that implements $f$ in mixed-strategy Bayesian Nash equilibrium
for every bounded quasilinear utility profile satisfying that bound. The
mechanism may depend on the bound, but not on the realized utility profile.
\end{corollary}

\begin{proof}
When $f$ satisfies NPD, the construction in the sufficiency proof chooses
all transfer magnitudes using only the finite primitives and the commonly
known bound on utility differences. Claims \ref{IMPLR} and \ref{PUREBNER}
then show that this single mechanism has a pure BNE and implements $f$ for
every bounded utility profile in the maintained domain, with zero transfers
on the equilibrium path.
\end{proof}

\begin{remark}[Limited Liability Transfers]
\label{COR_LIM_LIA} For every $\delta>0$, the mechanism can be chosen so
that 
\begin{equation*}
|\tau_i(m)|\leq 1+\delta \qquad\forall i\in\mathcal{I},\ m\in M.
\end{equation*}
Indeed, the sum of the transfers other than the cross-check, $%
\tau_i^1+\tau_i^2+\tau_i^4+\tau_i^5$, is proportional to $\varepsilon $ and
therefore converges uniformly to zero as $\varepsilon $ converges to zero.
The choice following \eqref{eq:crosscheck-scaleR} makes $\kappa $ converge
to one. Since the cross-check transfer is either zero or $-\kappa $,
choosing $\varepsilon $ sufficiently small gives the stated bound.
\end{remark}

\section{Uncertain State Types}

\label{NCKS}

In this section, we relax the common knowledge assumption about the state
and allow agents to be uncertain about one another's state types. We show
how the characterization extends to this setting under two standard
informational restrictions: nonexclusive information and information
smallness.

\subsection{Environment}

For each agent $i$, let $S_{i}$ be a finite set of state types. A state is a
profile $s=(s_{i})_{i\in \mathcal{I}}\in S=\prod_{i\in \mathcal{I}}S_{i}$,
and the SCF is a map $f:S\rightarrow A$. Nature draws $s$ according to a
commonly known prior $\pi \in \Delta (S)$, and agent $i$ observes only $%
s_{i} $. Without loss, take each $S_i$ to be the marginal support of $\pi$.
The baseline is embedded by letting $T$ denote the baseline state space,
setting $S_{i}=T$ for every $i$, and taking $\pi $ to be supported on
profiles satisfying $s_{1}=\cdots =s_{I}\in T$, with $\pi (t,\ldots ,t)>0$
for every $t\in T$.

Following the baseline model, we assume (1) the distribution of agent $i$'s
evidence endowment depends on $s_{i}$ but not on $s_{-i}$; and (2) evidence
draws are conditionally independent across agents, with $q_{j}^{e}(\cdot
\mid s_{j})\in \Delta (\mathbb{E}_{j})$. Define $p^{e}(s)(E)=\prod_{i\in 
\mathcal{I}}q_{i}^{e}(E_{i}\mid s_{i})$ and similarly $p_{-i}^{e}(s)(E)=%
\prod_{j\neq i}q_{j}^{e}(E_{j}\mid s_{j})$. Agent $i$'s type is a pair $%
(s_i,E_i)\in S_i\times\mathbb{E}_i$ with $q_i^e(E_i\mid s_i)>0$. For every
such pair, 
\begin{equation*}
\Pr (s_{-i},E_{-i}\mid s_{i},E_{i})=\pi (s_{-i}\mid s_{i})\prod_{j\neq
i}q_{j}^{e}(E_{j}\mid s_{j}).
\end{equation*}%
The evidence sets $\mathbb{E}_{i}$ and normality are as in Section~\ref%
{Model}: an agent endowed with $E_{i}$ may submit any $F_{i}\subseteq E_{i}$%
. Hence, evidence is ancillary: an agent's own evidence affects the set of
feasible evidence reports but does not change beliefs about opponents. All
interim belief objects below can therefore be indexed by state types alone.

An agent's state type determines higher-order beliefs about evidence. At the
first order, these are beliefs about the opponents' evidence endowments. At
each subsequent order, they are joint beliefs about the opponents' evidence
and their lower-order beliefs about evidence. The recursive construction is
given in Appendix~\ref{app:UST}. We assume that the state-type space admits
no redundancy, i.e.,\ for each agent, any two distinct state types differ in
their higher-order beliefs about evidence at some finite order.\footnote{%
See for example \cite{MZ1985}, \cite{ElyPeski2006}, and \cite{DFM2007}.} We
call $\bigl((S_i)_{i\in\mathcal{I}},\pi,(q_i^e)_{i\in\mathcal{I}}\bigr)$ an 
\emph{uncertain-state-type environment} if evidence is ancillary and
conditionally independent across agents given $s$, and the state-type spaces
satisfy no redundancy.

For any $s_{-i}\in\operatorname{proj}_{-i}(\supp\pi)$, let $\Pi _{i}(t_{i}\mid
s_{-i})$ be the posterior probability of agent $i$'s state type. The
information structure has \emph{information size} $\varepsilon _{I}$ if, for
every $i$ and every such $s_{-i}$, there is a unique state type $\chi
_{i}(s_{-i})$ satisfying 
\begin{equation*}
\Pi _{i}(\chi _{i}(s_{-i})\mid s_{-i})\geq 1-\varepsilon _{I},
\end{equation*}%
and, for every $s_{i}\in S_i$, 
\begin{equation}
\Pr \bigl(\chi _{i}(s_{-i})=s_{i}\mid s_{i}\bigr)\geq 1-\varepsilon _{I}.
\label{ISCOND}
\end{equation}%
For $s_{-i}\notin\operatorname{proj}_{-i}(\supp\pi)$, define $\chi _{i}(s_{-i})$
arbitrarily. A sequence of such environments is\emph{\ informationally small}
(\cite{MP2002,MP2004}) if $\varepsilon _{I}\rightarrow 0$ as $I\rightarrow
\infty $. In the present formulation with state-type profiles, \emph{%
nonexclusive information} is the special case $\varepsilon _{I}=0$. Thus the
opponents' state types determine $s_{i}$ exactly under nonexclusive
information, while under information smallness the selector $\chi
_{i}(s_{-i})$ predicts $s_{i}$ with conditional error probability at most $%
\varepsilon _{I}$.

\subsection{Generalized No Perfect Deceptions}

For $s\in \supp\pi $, let $\bar{S}(s)$ denote the \emph{belief closure} of $%
s $: the smallest subset of $\supp\pi $ containing $s$ and closed under
positive-probability interim beliefs. Thus, if $\bar{s}\in \bar{S}(s)$ and $%
\pi (t_{-i}\mid \bar{s}_{i})>0$, then $(\bar{s}_{i},t_{-i})\in \bar{S}(s)$.
Let $\bar{S}_{i}(s)=\{t_{i}\in S_{i}:(t_{i},t_{-i})\in \bar{S}(s)\text{ for
some }t_{-i}\}$. These closures partition $\supp\pi$, their projections
partition each $S_i$, and BNE conditions decompose across them. Thus a
copied equilibrium beginning at $s$ is modified throughout this component
and left unchanged on every other component.

An evidence report $F_{i}$ \emph{refutes} $r$, written $F_{i}\downarrow r$,
if no $E_i^{\prime }\in\mathbb{E}_i$ with $q_i^e(E_i^{\prime }\mid r_i)>0$
contains $F_i$. The profile $r$ is \emph{refutable at} $s$ if some agent has
an $E_i\in\mathbb{E}_i$ with $q_i^e(E_i\mid s_i)>0$ satisfying $%
E_i\downarrow r$; otherwise it is \emph{nonrefutable at} $s$.

A deception for agent $i$ is a reporting rule $\alpha _{i}:S_i\times \mathbb{%
E}_i\rightarrow S_{i}\times \Delta (\mathbb{E}_{i})$ such that, at every $%
(s_i,E_i)\in S_i\times\mathbb{E}_i$, its evidence-report distribution
assigns positive probability only to subsets of $E_i$. As in the baseline
model, we identify each element of $S_i\times \Delta(\mathbb{E}_i)$ with the
corresponding lottery on $S_i\times\mathbb{E}_i$ that is degenerate in its
state coordinate. We call each pair $(r_i,F_i)$ in the resulting support a 
\emph{pseudo type}. The value of the deception at $(s_i,E_i)$ is immaterial
when $q_i^e(E_i\mid s_i)=0$. The state coordinate is deterministic
conditional on $(s_i,E_i)$ but may depend on $E_i$. Whenever that coordinate
is the same for every $E_i$ with $q_i^e(E_i\mid s_i)>0$, denote its common
value by $\alpha_i(s_i)$ and write $\alpha(s)=(\alpha_i(s_i))_{i\in\mathcal{I%
}}$.

For each $(\tilde s_{-i},\tilde E_{-i})\in S_{-i}\times\mathbb{E}_{-i}$, let 
\begin{equation*}
\alpha_{-i}(\tilde s_{-i},\tilde E_{-i}) =\bigotimes_{j\neq
i}\alpha_j(\tilde s_j,\tilde E_j).
\end{equation*}
Given a deception $\alpha $, define the law of the opponents' pseudo types
faced by state type $\bar{s}_{i}$ as 
\begin{equation*}
\Lambda _{i}^{\alpha }(\hat{s}_{-i},\hat{E}_{-i}\mid \bar{s}_{i})={}\sum_{%
\tilde{s}_{-i}\in S_{-i}}\sum_{\tilde{E}_{-i}\in \mathbb{E}_{-i}}\pi (\tilde{%
s}_{-i}\mid \bar{s}_{i})p_{-i}^{e}(\bar{s}_{i},\tilde{s}_{-i})(\tilde{E}%
_{-i})\times \alpha _{-i}(\tilde{s}_{-i},\tilde{E}_{-i})(\hat{s}_{-i},\hat{E}%
_{-i}).
\end{equation*}%
For a state type $t_{i}$, let 
\begin{equation*}
\Lambda _{i}^{\ast }(\hat{s}_{-i},\hat{E}_{-i}\mid t_{i})=\pi (\hat{s}%
_{-i}\mid t_{i})p_{-i}^{e}(t_{i},\hat{s}_{-i})(\hat{E}_{-i})
\end{equation*}%
be the truthful joint law of the opponents' types when they submit maximal
evidence. The deception is \emph{perfect on} $\bar{S}(s)$ if, for every $i$,
every $\bar{s}_{i}\in \bar{S}_{i}(s)$, and every $E_i\in\mathbb{E}_i$ with $%
q_i^e(E_i\mid\bar s_i)>0$, 
\begin{equation*}
\Lambda _{i}^{\alpha }(\cdot \mid \bar{s}_{i})=\Lambda _{i}^{\ast }(\cdot
\mid t_i)\qquad\text{for every }(t_i,F_i)\in\supp\alpha_i(\bar s_i,E_i).
\end{equation*}%
The higher-order evidence-separation assumption makes the map $t_i\mapsto
\Lambda_i^*(\cdot\mid t_i)$ one-to-one: equality of these joint laws implies
equality of every order of evidence beliefs because opponents' lower-order
beliefs are functions of their state types. Because the left-hand side is
independent of $E_i$, perfectness therefore makes the state coordinate
identical across all $E_i$ satisfying $q_i^e(E_i\mid\bar s_i)>0$. Thus $%
\alpha_i(\bar s_i)$ is well defined, and each deceiving state type faces
exactly the strategic environment associated with the state type it reports.

\begin{definition}
An SCF $f$ satisfies \textbf{Generalized No Perfect Deceptions (GNPD)} if,
for every $s\in \supp\pi $ and every deception $\alpha$ that is perfect on $%
\bar S(s)$, the following implication holds. If, for every $\bar{s}\in \bar{S%
}(s)$, $\alpha(\bar{s})\in \supp\pi $ and is nonrefutable at $\bar{s}$, then 
$f(\alpha(s))=f(s).$
\end{definition}

If all relevant state types can jointly imitate another state-type profile,
if the deception makes each agent face the same interim law of opponents'
pseudo types as truthful reporting by the corresponding types would
generate, and if no evidence available anywhere in the belief closure can
expose it, then an implementing mechanism cannot assign a different outcome.

In the baseline model, the belief closure of each state is a singleton, and
perfectness then requires feasible withholding rules that reproduce $%
q_{i}^{e}(\cdot \mid \alpha _{i}(s_{i}))$ for every agent. Conditional
independence turns these joint-law requirements into the baseline
perfect-deception equations. Hence, GNPD reduces to NPD.

\subsection{Implementation}

The mechanisms in this section may contain continuous reflection reports.%
\footnote{%
Remark~\ref{rem:continuous-reflection} in Appendix~\ref{app:UST} explains
why the finite reflection-report domain used in the baseline does not carry
over to the present construction.} For a mechanism $\mathcal{M}=(M,g,\tau )$%
, each $M_{i}$ is a compact metric space, $M=\prod_{i}M_{i}$, and $\Delta
(M_{i})$ denotes the set of Borel probability measures on $M_{i}$. For each
evidence endowment $E_{i}$, the feasible message set $\varnothing \neq
M_{i}(E_{i})\subseteq M_{i}$ is closed, and evidence feasibility is nested: $%
E_{i}^{\prime }\subseteq E_{i}$ implies $M_{i}(E_{i}^{\prime })\subseteq
M_{i}(E_{i})$. For every bounded utility profile, the induced payoff
functions are bounded and continuous in messages. In particular, they are
Borel measurable.

A mixed strategy for agent $i$ is a map 
\begin{equation*}
\sigma _{i}:S_i\times\mathbb{E}_i\longrightarrow \Delta (M_{i})
\end{equation*}%
such that $\sigma _{i}(s_{i},E_{i})(M_{i}(E_{i}))=1$ for every $(s_i,E_i)\in
S_i\times\mathbb{E}_i$. A pure strategy is a degenerate mixed strategy. Given
a strategy profile $\sigma $, the expected
utility from a feasible message $m_{i}$ for type $(s_{i},E_{i})$ is 
\begin{align*}
U_{i}(m_{i},\sigma _{-i}\mid s_{i},E_{i})={}&
\sum_{s_{-i}\in S_{-i}}\pi(s_{-i}\mid s_{i}) \\
& \quad \times \sum_{E_{-i}\in \mathbb{E}_{-i}}p_{-i}^{e}\left(
s_{i},s_{-i}\right) (E_{-i}) \\
& \quad \times \int_{M_{-i}(E_{-i})}\Bigl[ \bar{u}%
_{i}(g(m_{i},m_{-i}),(s_{i},s_{-i})) \\
& \hspace{9em}{}+\tau _{i}(m_{i},m_{-i})\Bigr]
\,d\sigma _{-i}(m_{-i}\mid s_{-i},E_{-i}),
\end{align*}%
where $M_{-i}(E_{-i})=\prod_{j\neq i}M_{j}(E_{j})$ and $\sigma _{-i}(\cdot
\mid s_{-i},E_{-i})=\bigotimes_{j\neq i}\sigma _{j}(s_{j},E_{j})$ is the
opponents' product measure. A strategy profile is a BNE if, for every $i\in 
\mathcal{I}$, 
\begin{equation*}
\begin{aligned}
\supp\sigma _{i}(s_{i},E_{i})&\subseteq \arg \max_{m_{i}\in
M_{i}(E_{i})}U_{i}(m_{i},\sigma _{-i}\mid s_{i},E_{i}),\\
&\hspace{-6em}\forall (s_i,E_i)\in S_i\times\mathbb{E}_i\text{ such that }
q_i^e(E_i\mid s_i)>0.
\end{aligned}
\end{equation*}

For every state $s\in \supp\pi $, the equilibrium strategy profile induces a
Borel probability measure $\mu _{\sigma }^{s}$ on $M$. For every Borel set $%
\bar{M}\subseteq M$, 
\begin{equation}
\mu _{\sigma }^{s}(\bar{M})=\sum_{E\in \mathbb{E}}p^{e}(s)(E)\left[
\prod\limits_{i\in \mathcal{I}}\sigma _{i}(s_{i},E_{i})\right] (\bar{M}).
\label{MUST}
\end{equation}%
An SCF is exactly implementable in mixed-strategy BNE if, for every bounded
utility profile, the induced Bayesian game has a pure-strategy BNE and every
BNE satisfies
\begin{equation*}
g(m)=f(s)\quad \text{and}\quad \tau _{i}(m)=0\quad \forall i\in \mathcal{I}
\end{equation*}%
for $\mu _{\sigma }^{s}$-almost every $m$ at every $s\in \supp\pi $.

For the transfer conclusion, define 
\begin{equation*}
\mathbb{E}_{\sigma }[|\tau _{i}|]=\sum_{s\in \supp\pi }\pi (s)\int_{M}|\tau
_{i}(m)|\,d\mu _{\sigma }^{s}(m).
\end{equation*}%
A sequence of mechanisms \emph{exactly implements outcomes with vanishing
expected transfers} if, for sufficiently large $I$ and every bounded utility
profile, the induced game has a pure-strategy BNE, every BNE selects $f(s)$
for $\mu _{\sigma }^{s}$-almost every message at every $s\in \supp\pi $ and
\begin{equation}
\sup_{\bar{u}}\ \sup_{\sigma \in \mathrm{BNE}(\mathcal{M}_{I},\bar{u}%
)}\max_{i\in \mathcal{I}}\mathbb{E}_{\sigma }[|\tau _{i}|]\longrightarrow 0,
\label{VANTR}
\end{equation}%
where the first supremum ranges over the normalized bounded utility domain.
The outcome conclusion is exact for sufficiently large $I$, and the expected
transfers converge to zero.

\begin{theorem}
\label{THM_GTS}Consider any uncertain-state-type environment.

\begin{enumerate}
\item If agents are informationally small, with $\varepsilon _{I}\rightarrow
0$ along a sequence of environments, then every SCF satisfying GNPD admits
exact outcome implementation in mixed-strategy BNE by mechanisms which
exactly implement with vanishing expected transfers.

\item If information is nonexclusive ($\varepsilon _{I}=0$), an SCF is
exactly implementable in mixed-strategy BNE if and only if it satisfies GNPD.
\end{enumerate}
\end{theorem}

The formal proof of Theorem \ref{THM_GTS} appears in Appendix \ref{app:UST},
and we provide a brief summary of the differences from the baseline here.
The sufficiency proof retains the baseline sequence: obtain consistent
reports, eliminate refutable and imperfect nonrefutable lies, and invoke
GNPD in place of NPD. Unlike in the baseline, consistency here requires
iterated reflection reports to propagate alarms across the belief closure;
higher-order evidence beliefs identify privately observed state types, and
challenges use state-type-contingent joint laws of pseudo types. Transfers
are zero under nonexclusive information and vanish in expectation under
information smallness. For every bounded utility profile, the construction
also yields a pure BNE: agents truthfully report their state types and
evidence hierarchies, submit maximal evidence, file no challenges, and make
the conditional-mean reflection reports specified by the scoring rules.

\section{Related Literature}

\label{RL}

In this section, we discuss our paper's results alongside some other related
papers to provide the reader a comparison of our results against the
literature. First, this paper's primary contribution is to the literature on
implementation with hard evidence, pioneered by \cite{BL2012} and \cite%
{KT2012} and extended by \cite{BCS2020}. We make two main contributions
here: we incorporate uncertainty both in the evidence dimension and, to a
limited but nontrivial extent, in the payoff-relevant state types. \cite%
{P2019} studies a variation of the implementation problem with uncertainty
about evidence along with full Bayesian uncertainty in the state dimension
using direct mechanisms, pure deceptions, and a transfer-free environment.
We study indirect mechanisms, deceptions that mix over evidence reports, and
monetary transfers. Our baseline message spaces are finite, and the message
spaces in the uncertain-state-type extension are compact.\footnote{%
Allowing deceptions to mix over evidence reports in \cite{P2019} requires a
strengthening of the evidence distinguishability condition therein if the
designer wishes only to use direct mechanisms. We note here that the
analogue in our context---the necessary condition for direct implementation---
remains an open problem.} The article establishes Bayesian Evidence
Monotonicity (BEM) as the necessary condition for implementation. We study a
rich class of indirect mechanisms with finite or compact message spaces.

With uncertain evidence, perfect deceptions make NPD the relevant
strengthening of the measurability condition in \cite{BL2012}. Both beliefs
about evidence distributions and realized evidence endowments matter for
implementation. Our mechanisms construct test allocations
using the Farkas lemma, use scoring rules to elicit beliefs about evidence
distributions, and use reflection transfers both to turn refutations into
full-disclosure incentives in the baseline and to transmit alarms through
belief closures in the extension. The associated challenge and scoring
transfers vanish at the required equilibria.\footnote{%
We conjecture that techniques similar to the ones outlined in this paper may
find use in tackling the problem of creating bounded mechanisms in the style
of \cite{MAM} for the classical incomplete information implementation
problem studied in \cite{J1991}, potentially under assumptions such as
nonexclusive information or information smallness.}

Handling these challenges jointly yields a mechanism that is not a
practically simple procedure, as it uses state reports that index evidence
distributions, scoring rules, cross-checks, reflection reports, and
challenge transfers. Designing natural and bounded mechanisms for specific
settings or social choice rules, as \cite{EN2025} does for efficient
implementation, is an important future direction for which the techniques
developed here may prove useful.

Second, in the classical Bayesian implementation literature
pioneered in \cite{J1991} and later developed in \cite{SV2010}, our main
contribution lies in extending these results to include evidence, allowing
for implementation with state-independent preferences. This enables
solutions to important practical problems such as litigation and budget
allocation. Our state-type extension uses nonexclusive information or
informational smallness. The mechanisms in \cite{J1991} use modulo games,
and those in \cite{SV2010} use integer games. Our baseline result achieves
Bayesian implementation in a finite mechanism while explicitly accounting
for mixed-strategy equilibria.

Third, the information-smallness literature (e.g., \cite{MP2002,MP2004})
studies when an agent has little influence on posterior beliefs and when the
resulting incentive costs can be made small. Our information smallness
assumption is a specialization suited to an SCF defined on the full
state-type profile: the opponents' state types must predict the omitted
state type with error $\varepsilon _{I}\rightarrow 0$. Given that
prediction, the challenge transfers and finitely many higher-order beliefs
about evidence identify state types after an alarm, preserving exact
outcomes while making expected equilibrium transfers small.

\cite{BDL2021} provide an excellent survey of the mechanism design with
evidence literature. An early contribution to mechanism design with
evidence is \cite{GL86}, who study the principal-agent problem when the
agent cannot manipulate the truth arbitrarily. This corresponds to a notion
of evidence---an agent, by presenting messages which are only available to
her in a certain set of states, can prove that the true state is within that
set. Early contributions in the field of game theory and evidence include 
\cite{M81}, \cite{G81} and \cite{D85}. More papers involving mechanism
design and evidence include \cite{BW2007}, \cite{DS2008}, \cite{HKP2017},
and \cite{BDL2019}. \cite{BDL2021} study a setup involving stochastic
evidence. While we assume that agents are endowed with evidence to begin
with, \cite{BDL2021} study the process of evidence acquisition and the value
of commitment. In addition, they focus on partial implementation whereas our
focus is on full implementation.

\bibliographystyle{plainnat}
\bibliography{references}

\appendix

\section{How Restrictive is NPD?}

\label{app:VERIFYNPD}

This appendix proves Lemma~\ref{lem:perfect-deception-characterization} and
records simple classes of evidence distributions for which NPD can be
checked directly.

\begin{proof}[Proof of Lemma~\protect\ref%
{lem:perfect-deception-characterization}]
The equivalence of \textup{(i)} and \textup{(ii)} follows from the coupling
characterization of the usual stochastic order; see %
\citet{ShakedShanthikumar2007}. Indeed, $s\rightarrow_i s^{\prime}$ holds if
and only if there is a feasible withholding rule that transforms an
endowment distributed according to $p_i(s)$ into a report distributed
according to $p_i(s^{\prime})$. Equivalently, there is a coupling $E\sim
p_i(s)$, $F\sim p_i(s^{\prime})$ such that $F\subseteq E$ almost surely. The
coupling characterization says that such a coupling exists if and only if %
\eqref{eq:upper-set-dominance} holds.

Next, suppose \textup{(i)} holds. Then, by definition, $p_i(s^{\prime})\in 
\mathcal{P}_i(s,s^{\prime})$. Consequently, condition \textup{(iii)} holds:
if such a vector existed, taking $q=p_i(s^{\prime})$ would require both $%
b\cdot p_i(s^{\prime})<0$ and $b\cdot p_i(s^{\prime})>0$. Thus \textup{(i)}
implies \textup{(iii)}.

Conversely, suppose \textup{(i)} fails. By the equivalence of \textup{(i)}
and \textup{(ii)}, there is an upward-closed set $\mathcal{U}\subseteq%
\mathbb{E}_i$ such that $p_i(s^{\prime})(\mathcal{U})>p_i(s)(\mathcal{U})$.
Let $c=\bigl(p_i(s^{\prime})(\mathcal{U})+p_i(s)(\mathcal{U})\bigr)/2$ and
define $b(E)=c-\mathbf{1}_{\{E\in\mathcal{U}\}}$ for $E\in\mathbb{E}_i$.
Then $b\cdot p_i(s^{\prime})=c-p_i(s^{\prime})(\mathcal{U})<0$. Now let $q\in%
\mathcal{P}_i(s,s^{\prime})$. Since $q$ is generated from $p_i(s)$ by
withholding evidence and $\mathcal{U}$ is upward closed, $q(\mathcal{U})\leq
p_i(s)(\mathcal{U})$. Therefore $b\cdot q=c-q(\mathcal{U})\geq c-p_i(s)(%
\mathcal{U})>0$. Hence $b$ satisfies the inequalities excluded by \textup{%
(iii)}, so \textup{(iii)} fails. This proves the contrapositive of \textup{%
(iii)} implying \textup{(i)}, and completes the proof.
\end{proof}

The lemma turns NPD into a finite set of checks. For any ordered pair $%
(s,s^{\prime})$ with $f(s)\neq f(s^{\prime})$, NPD is automatically
satisfied if $s^{\prime}$ is refutable at $s$. If $s^{\prime}$ is
nonrefutable at $s$, NPD holds for the pair if and only if some agent $i$
violates \eqref{eq:upper-set-dominance}; equivalently, there are an agent $i$
and an upward-closed $\mathcal{U}\subseteq\mathbb{E}_i$ such that $%
p_i(s^{\prime})(\mathcal{U})>p_i(s)(\mathcal{U})$.

The test simplifies to a one-dimensional comparison when the relevant
distributions lie on a chain. Suppose, for a nonrefutable pair $%
(s,s^{\prime})$, that some agent $i$ satisfies $\supp p_i(s)\cup\supp %
p_i(s^{\prime})\subseteq\{E^0,\ldots,E^K\}$ with $E^0\subsetneq
E^1\subsetneq\cdots\subsetneq E^K$. Then each upward-closed set meets the
chain in an upper tail, so $s\rightarrow_i s^{\prime}$ if and only if $%
\sum_{\ell=k}^K p_i(s)(E^\ell)\geq\sum_{\ell=k}^K p_i(s^{\prime})(E^\ell)$
for every $k$. Thus agent $i$ blocks the deception exactly when the false
state assigns strictly more probability to some upper tail of the chain than
the true state does. When $K=1$, this says simply that the more informative
endowment is strictly more likely at $s^{\prime}$ than at $s$.

\section{Proof of Theorem \protect\ref{THM_GTS}}

\label{app:UST}

This appendix proves Theorem~\ref{THM_GTS}. Necessity keeps the definition
of GNPD in terms of state-type profiles and uses only ancillary evidence.
Sufficiency uses the formal construction of higher-order beliefs about
evidence given below. The same mechanism gives exact outcomes under both
information restrictions; its equilibrium transfers are exactly zero under
nonexclusive information and vanish in expectation under information
smallness. For the informationally small case, we restrict attention to
sufficiently large $I$, so that $\varepsilon _{I}<1/2$.

\subsection*{Formal construction of higher-order beliefs about evidence}

Let $X_{i}^{1}=\mathbb{E}_{-i}$ and define agent $i$'s order-one evidence
belief by 
\begin{equation*}
\mu _{i}^{1}(s_{i})=\mathcal{L}(E_{-i}\mid s_{i})\in \Delta (X_{i}^{1}).
\end{equation*}%
Write $\beta _{i}^{1}(s_{i})=\mu _{i}^{1}(s_{i})$. Recursively, after $\beta
_{j}^{r-1}$ has been defined for every agent $j$, let 
\begin{equation*}
X_{i}^{r}=\mathbb{E}_{-i}\times \prod_{j\neq i}\beta _{j}^{r-1}(S_{j})
\end{equation*}%
and define 
\begin{equation}
\mu _{i}^{r}(s_{i})=\mathcal{L}\bigl(E_{-i},\beta _{-i}^{r-1}(s_{-i})\mid
s_{i}\bigr)\in \Delta (X_{i}^{r}),\qquad \beta _{i}^{r}(s_{i})=\bigl(\mu
_{i}^{1}(s_{i}),\ldots ,\mu _{i}^{r}(s_{i})\bigr).  \label{HIER}
\end{equation}%
Thus order one predicts opposing evidence; order two predicts opposing
evidence and opponents' order-one evidence beliefs; and each further order
predicts opposing evidence jointly with the opponents' entire lower-order
beliefs about evidence.

The higher-order evidence-separation assumption in Section~\ref{NCKS}
requires distinct state types to differ at some finite order. Because each $%
S_i $ is finite, there is a common finite order $K$ such that $%
\beta_i^K:S_i\rightarrow\beta_i^K(S_i)$ is one-to-one for every agent. The
mechanism elicits the hierarchy only through this order.

\subsection*{Relation to the Baseline Proof}

Necessity is the baseline copy argument applied to each $\bar s_i\in\bar
S_i(s)$ and $E_i\in\mathbb{E}_i$ satisfying $q_i^e(E_i\mid\bar s_i)>0$:
generalized perfectness matches the joint laws of opponents' pseudo types,
while nonrefutability and nested message feasibility transfer deviations to
the corresponding pseudo types.

For sufficiency, the correspondence is claim by claim.

\begin{enumerate}
\item Lemma~\ref{lem:USTmenu} generalizes the separating challenge in 
\textup{(\ref{eq:challenge-separation})}: a finite menu separates
conditional joint laws of pseudo types across all feasible deceptions.

\item Lemma~\ref{lem:USTreflection} strengthens Claim~\ref{mainR}: an alarm
propagates across the belief closure and yields maximal evidence, truthful
reports of higher-order beliefs about evidence, and truthful state-type
reports.

\item Claims~\ref{cl:USTcons} and \ref{cl:USTpure} perform the work of Claim~%
\ref{CONSISTENCYR}, showing that equilibrium state-type first-report
profiles belong to $\supp\pi$ and define a single state coordinate for each
state type, independent of its realized evidence.

\item Claim~\ref{cl:USTnoref} is the belief-closure counterpart of Claim~\ref%
{NOORLR}, excluding refutable profiles throughout $\bar S(s)$.

\item Using GNPD and Lemma~\ref{lem:USTmenu}, Claim~\ref{cl:USTwrong}
replaces Claim~\ref{NONRLR} and the outcome part of Claim~\ref{IMPLR}. The
final transfer argument completes the correspondence: transfers are zero
under nonexclusive information, whereas expected transfers vanish under
information smallness.
\end{enumerate}

\subsection*{Necessity}

Suppose that a mechanism exactly implements $f$. Fix state-independent
utilities $\bar u_i(a,s)=v_i(a)$, a state $s\in\supp\pi$, and a deception $%
\alpha $ satisfying the antecedents of GNPD on $\bar S(s)$. Let $\sigma$ be
any BNE of the mechanism. For every $\bar s_i\in\bar S_i(s)$ and every $%
E_i\in\mathbb{E}_i$ with $q_i^e(E_i\mid\bar s_i)>0$, define $%
\widehat\sigma_i(\bar s_i,E_i)$ by drawing the pseudo type $(\alpha_i(\bar
s_i),F_i)$ according to $\alpha_i(\bar s_i,E_i)$ and then playing the
corresponding type's strategy $\sigma_i(\alpha_i(\bar s_i),F_i)$. This is
feasible by evidence-feasibility nesting because $F_i\subseteq E_i$.
Moreover, whenever this pseudo type is drawn with positive probability,
perfectness implies $q_i^e(F_i\mid \alpha_i(\bar s_i))>0$: place $\bar s_i$
in a positive-prior profile in the closure and apply an opponent's
perfectness equation. Thus every pseudo type drawn can occur conditional on
its state type.

Fix $\bar s_i\in\bar S_i(s)$ and $E_i\in\mathbb{E}_i$ satisfying $%
q_i^e(E_i\mid\bar s_i)>0$, and write $r_i=\alpha_i(\bar s_i)$. Perfectness
gives this type under $\widehat\sigma$ exactly the opponents' message law
faced by state type $r_i$ under $\sigma$. Nonrefutability gives an endowment 
$E_i^{\prime }\in\mathbb{E}_i$ with $E_i^{\prime }\supseteq E_i$ and $%
q_i^e(E_i^{\prime }\mid r_i)>0$, so every deviation feasible for $E_i$ is
feasible for $(r_i,E_i^{\prime })$. By ancillarity, exact implementation,
and state-independent utilities, every type with state type $r_i$ obtains
the same interim equilibrium payoff under $\sigma$, and each copied strategy
obtains that payoff as well. Hence a profitable deviation under $%
\widehat\sigma$ would be a profitable deviation for $(r_i,E_i^{\prime })$
under $\sigma$, a contradiction. Every pair $(\bar s_i,E_i)$ with $\bar
s_i\in\bar S_i(s)$ and $q_i^e(E_i\mid\bar s_i)>0$ therefore optimizes.
Retain $\sigma$ at all remaining pairs. Since BNE conditions decompose
across belief components, $\widehat\sigma$ is then a BNE of the full game.

Finally, every profile of pseudo types generated by $\alpha$ with positive
probability at $s$ has positive probability at the state $\alpha(s)\in \supp%
\pi$, by the positivity established above and conditional independence.
Conditional on any such profile, the copied message law is exactly the
equilibrium message law under $\sigma$ at $\alpha(s)$. Exact implementation
therefore yields the outcome $f(\alpha(s))$ under $\widehat\sigma$. Since $%
\widehat\sigma$ is also a BNE at $s$, exact implementation yields $f(s)$.
Thus $f(\alpha(s))=f(s)$, proving GNPD.

\subsection*{Sufficiency}

Fix an SCF satisfying GNPD. Let $K$ be the common finite separating order
fixed by higher-order evidence separation and finiteness of the state-type
spaces. We first construct a finite challenge menu and then give the
mechanism.

\subsubsection*{A finite challenge menu for laws of pseudo types}

\begin{lemma}
\label{lem:USTmenu} For every $C>0$ there is a finite collection of entries $%
(i,t_i,t_i^{\prime },b)$, with $b\in\mathbb{R}^{|S_{-i}\times\mathbb{E}%
_{-i}|}$, such that the following holds. Let $s\in\supp\pi$ and let $%
\varphi=(\varphi_i)_i$ be a pure state-type first-report map for which $%
f(\varphi(s))\ne f(s)$ and, for every $\bar s\in\bar S(s)$, $\varphi(\bar s)$
belongs to $\supp\pi$ and is nonrefutable at $\bar s$. For every feasible
deception $\alpha$ whose first coordinate at $(t_i,E_i)$ is $\varphi_i(t_i)$
for every $i$, $t_i\in S_i$, and $E_i\in\mathbb{E}_i$ satisfying $%
q_i^e(E_i\mid t_i)>0$, there are an agent $i$, a state type $t_i\in\bar
S_i(s)$, and a menu vector indexed by $(i,t_i,\varphi_i(t_i))$ satisfying 
\begin{equation}
b\cdot\Lambda_i^\alpha(\cdot\mid t_i) >b\cdot\Lambda_i^*(\cdot\mid
\varphi_i(t_i))+4C.  \label{MENUMARGIN}
\end{equation}
\end{lemma}

\begin{proof}
Fix $s$ and $\varphi$. The feasible evidence-garbling kernels with first
coordinate $\varphi_i(t_i)$ at every type form a compact product of
simplices, denoted $A(s,\varphi)$. For each $i,t_i$, the map $%
\alpha\mapsto\Lambda_i^\alpha(\cdot\mid t_i)$ is multi-affine and
continuous. Define 
\begin{equation*}
V(\alpha)=\max_{i}\max_{t_i\in\bar S_i(s)} \left\|\Lambda_i^\alpha(\cdot\mid
t_i) -\Lambda_i^*(\cdot\mid \varphi_i(t_i))\right\|.
\end{equation*}
If $V(\alpha)=0$, then $\alpha$ is perfect on $\bar S(s)$, and GNPD would
imply $f(\varphi(s))=f(s)$. Hence $V>0$ on the compact set $A(s,\varphi)$.

For each $\alpha\in A(s,\varphi)$ choose a pair $(i,t_i)$ attaining a
positive difference and a vector $b_\alpha$ that strictly separates the two
finite probability vectors. Rescale $b_\alpha$ so that the difference at $%
\alpha$ exceeds $4C$. Continuity gives an open neighborhood of $\alpha$ on
which the same strict inequality holds. These neighborhoods cover $%
A(s,\varphi)$; compactness gives a finite subcover. Taking the corresponding
finite set of vectors and then taking the union over the finitely many pairs 
$(s,\varphi)$ proves the lemma.
\end{proof}

\subsubsection*{Messages, alarms, and transfers}

For each $i$ and order $r$, let 
\begin{equation*}
\mathcal{Z}_{i}^{r}=\{\mu _{i}^{r}(t_{i}):t_{i}\in S_{i}\}
\end{equation*}%
be the finite set of canonical order-$r$ reports. A hierarchy report is a
tuple $z_{i}=(z_{i}^{1},\ldots ,z_{i}^{K})\in \prod_{r=1}^{K}\mathcal{Z}%
_{i}^{r}$. It is \emph{canonical for} $t_{i}$ if $z_{i}^{r}=\mu
_{i}^{r}(t_{i})$ for every $r$. By higher-order evidence separation and the
choice of $K$, at most one state type has a given canonical tuple. Define
the state-type-identification map $d_{i}(z_{i})$ to be that state type when
the tuple is canonical, and an arbitrary default otherwise. The state-type
prediction generated by opponents' hierarchy reports is 
\begin{equation}
\widehat{s}_{i}(z_{-i})=\chi _{i}\bigl(d_{-i}(z_{-i})\bigr).
\label{TYPEPRED}
\end{equation}%
When opponents report their true hierarchies, this is $\chi _{i}(s_{-i})$.

Since $\supp\pi$ is finite, repeatedly applying the positive-probability
condition in the definition of $\bar S(s)$ reaches the entire belief closure
after finitely many rounds. Fix $L$ greater by two than the maximum number
of rounds required across all belief closures. Because an agent's evidence
endowment does not affect her beliefs about opponents, the same $L$ applies
to every evidence type associated with the state types in the closure.

Let $C>0$ be a parameter to be chosen below and, for each chosen value, take
a menu satisfying Lemma~\ref{lem:USTmenu}. For each triple $%
(i,t_i,t_i^{\prime })$, enumerate its vectors as $b_i^1(t_i,t_i^{\prime
})$, $\ldots$, $b_i^{N_i(t_i,t_i^{\prime })}(t_i,t_i^{\prime })$, and let $\bar
N_i $ be the largest such number. Agent $i$ reports 
\begin{equation*}
m_i=(r_i,z_i^1,\ldots,z_i^K,F_i, c_i^1,\ldots,c_i^L,t_i^b,\ell_i^b),
\end{equation*}
where $r_i,t_i^b\in S_i$, $z_i^r\in\mathcal{Z}_i^r$, $F_i\in\mathbb{E}_i$, $%
c_i^\ell\in[0,1]$, and $\ell_i^b\in\{0,\ldots,\bar N_i\}$. Evidence
feasibility requires $F_i\subseteq E_i$.

The outcome is 
\begin{equation*}
g(m)= 
\begin{cases}
f(r), & r\in\supp\pi, \\ 
a^0, & r\notin\supp\pi,%
\end{cases}%
\end{equation*}
where $a^0$ is fixed. An active challenge is filed when $1\le\ell_i^b\le
N_i(t_i^b,r_i)$.

Let $D(m)=1$ if at least one of the following occurs:

\begin{enumerate}
\item $r\notin\supp\pi$;

\item for some $i$, $z_i$ is not canonical for $r_i$;

\item $r_i\ne\widehat s_i(z_{-i})$ for some $i$;

\item $F_i\downarrow r$ for some $i$;

\item some agent files an active challenge.
\end{enumerate}

Otherwise let $D(m)=0$. Thus an alarm occurs whenever $r\notin\supp\pi$,
whenever a hierarchy is inconsistent with a state-type report, or whenever a
state-type report differs from the state type predicted by opponents'
hierarchy reports; refutations and challenges are also alarms. Such a
mismatch can occur on the truthful path under information smallness; this is
harmless for outcomes and will be rare in transfers.

Define 
\begin{equation*}
\omega_i(m)=\max\left\{D(m), \max_{j\ne i}\max_{1\le\ell\le
L}c_j^\ell\right\}.
\end{equation*}
The unscaled evidence reward is 
\begin{equation*}
\tau_i^E(m)=a_E\,\omega_i(m)|F_i|,
\end{equation*}
where $a_E>0$ is chosen below.

Reflection reports forecast an alarm and then forecast positive lower-order
reflection values. Let 
\begin{equation*}
R_i^1(m)=D(m), \qquad R_i^\ell(m)=\max_{j\ne i}c_j^{\ell-1}
\quad(\ell=2,\ldots,L).
\end{equation*}
For $c,x\in[0,1]$, let $S^R(c,x)=2cx-c^2$. This score uniquely elicits the
conditional mean of $x$. The unscaled reflection transfer is 
\begin{equation*}
\tau_i^R(m)=\sum_{\ell=1}^L S^R(c_i^\ell,R_i^\ell(m)).
\end{equation*}
The continuous realization $\max_{j\ne i}c_j^{\ell-1}$ makes the transfer
continuous and retains the propagation property: a positive realization with
positive probability has a strictly positive conditional mean.

\begin{remark}[Why reflection reports are continuous]
\label{rem:continuous-reflection} In the baseline proof, consistency of the
state reports is established before the reflection argument is used.
Conditional on a consistent profile, a reflection report forecasts whether a
fixed state claim will be refuted. Conditional independence pins this
forecast down to a primitive refutation probability $r_i(s,s^{\prime})$, so
the report domain can be restricted to the finite set $C^R$.

The present construction instead uses reflection as part of the argument
establishing consistency of the state-type first reports. At order one, the
quadratic score elicits the conditional probability of the alarm $D$, which
depends on equilibrium mixing. At each higher order, $c_{i}^{\ell }$ elicits
the conditional mean of $\max_{j\neq i}c_{j}^{\ell -1}$. These conditional
means need not belong to any finite set determined by the primitives. On a
fixed finite report domain containing zero, a sufficiently small positive
conditional mean would be optimally reported as zero under the quadratic
score, breaking the positivity implication used in Lemma~\ref%
{lem:USTreflection} to propagate an alarm through a belief-closed component.
Allowing $c_{i}^{\ell }\in \lbrack 0,1]$ ensures that every conditional mean
can be reported exactly and preserves the propagation argument.
\end{remark}

We next score the evidence hierarchy. Let $Q(q,x)=2q(x)-q\cdot q$ be the
quadratic score on a finite outcome set. For $r=1$, set $Y_i^1=F_{-i}$. For $%
r\ge2$, set 
\begin{equation*}
Y_i^r=(F_{-i},z_{-i}^1,\ldots,z_{-i}^{r-1}).
\end{equation*}
The canonical distribution $\mu_i^r(t_i)$ is extended by zero to tuples of
lower-order reports that are not jointly canonical, so it can be evaluated
at every $Y_i^r$. Define 
\begin{equation}
\tau_i^H(m)=\sum_{r=1}^K\lambda_r \left[ Q(z_i^r,Y_i^r) -Q\bigl(%
\mu_i^r(\widehat s_i(z_{-i})),Y_i^r\bigr) \right].  \label{HIERTRANSFER}
\end{equation}
The offset is independent of agent $i$'s own hierarchy reports. If the
state-type report, hierarchy report, and state-type prediction agree, the
two scores in every bracket are identical and cancel pointwise.

The unscaled challenge transfer is zero if the challenge is inactive. If it
is active, let 
\begin{align}
\tau_i^B(m)={}& b_i^{\ell_i^b}(t_i^b,r_i)(r_{-i},F_{-i})
-b_i^{\ell_i^b}(t_i^b,r_i) \cdot\Lambda_i^*(\cdot\mid r_i)-C.
\label{CHALLTRANS}
\end{align}
It is a bet on opponents' pseudo types, centered at the truthful law of the
opponents' types conditional on the challenger's state-type report.

The cross-check transfer is 
\begin{equation}
\tau_i^C(m)=-\rho_I \mathbf{1}\{r_i\ne\widehat s_i(z_{-i})\}.
\label{CROSSTRANS}
\end{equation}
The total transfer is 
\begin{equation}
\tau_i(m)=\eta_I \bigl[\tau_i^E(m)+\tau_i^R(m)+\tau_i^H(m)+\tau_i^B(m)\bigr] %
+\tau_i^C(m).  \label{TOTALUST}
\end{equation}

We choose constants in the following order. Fix any $a_E>0$. Replacing a
nonmaximal evidence report by maximal evidence weakly raises $D$ and hence $%
\omega_i$; if it activates a refutation, it also weakly raises the first
reflection score. Maximal evidence is therefore strictly optimal whenever a
type assigns positive probability to $\omega_i(m)>0$. Next choose each $%
\lambda_r$ so that, once maximal evidence and truthful lower-order reports
make the order-$r$ realization truthful, the strict quadratic-score gain
from reporting $\mu_i^r(s_i)$ dominates every other unscaled payoff effect
of the order-$r$ report. Then choose $C$ larger than every possible gain in
the other unscaled auxiliary terms caused by activating or deactivating a
challenge, and use the menu with margin $4C$.

After these finite constants and the menu have been fixed, let $M_I<\infty$
bound the absolute value of the bracketed auxiliary transfer in \textup{(\ref%
{TOTALUST})} and all one-message changes in that bracket. Choose $\delta_I>0$%
; along an informationally small sequence, take $\delta_I\to0$. Set 
\begin{equation}
0<\eta_I\le\frac{\delta_I}{1+M_I}.  \label{ETASCALE}
\end{equation}
This scaling changes no strict comparison inside the auxiliary block and
ensures that its total absolute magnitude is at most $\delta_I$. Finally set 
\begin{equation}
\rho_I=\frac{1+3\delta_I}{1-2\varepsilon_I}.  \label{RHOSCALE}
\end{equation}
The scaled auxiliary transfer can change by at most $2\delta_I$ between two
messages, whereas an outcome-utility difference is less than one. Whenever
truthful hierarchy reports are elicited, reporting the true state type
rather than any fixed false state type reduces the cross-check probability
by at least $1-2\varepsilon_I$. Its expected cross-check gain is therefore
at least 
\begin{equation*}
\rho_I(1-2\varepsilon_I)=1+3\delta_I>1+2\delta_I,
\end{equation*}%
which strictly dominates every possible outcome and auxiliary-transfer gain.

This same choice gives both vanishing expected transfers and limited
liability. Since $|\tau_i^C|\le\rho_I$, the pointwise bound on the total
transfer is $|\tau_i|\le\rho_I+\delta_I$. If $\varepsilon_I\to0$ and $%
\delta_I\to0$, then 
\begin{equation*}
\rho_I\varepsilon_I\longrightarrow0, \qquad \rho_I+\delta_I\longrightarrow1.
\end{equation*}%
Consequently, under information smallness, for every $\kappa>0$ and
sufficiently large $I$, the following bound holds: 
\begin{equation*}
\sup_{i\in \mathcal{I}}\sup_{m\in M}|\tau _{i}(m)|<1+\kappa .
\end{equation*}
Under nonexclusive information, $\varepsilon_I=0$ and the same conclusion
holds for any fixed $I$ by choosing $\delta_I<\kappa/4$.

This is the formal version of the \emph{limited liability} claim -- the
mechanism only requires transfers of the order of the outcome to discipline
agents. When the common-knowledge benchmark also satisfies higher-order
evidence separation---that is, when distinct states induce distinct profiles 
$(p_i(s))_{i\in\mathcal{I}}$ of evidence laws---the present construction
gives the same conclusion. Its disadvantage relative to the baseline
construction is the more complicated infinite message structure.

All discrete report components have the discrete topology. The only
continuous components are the reflection reports. The alarm $D$ depends only
on discrete components, and the higher reflection realizations are maxima of
continuous reports. Hence all transfers are continuous on compact message
spaces. The induced finite-type Bayesian game therefore has a mixed-strategy
BNE.

\subsubsection*{A pure-strategy equilibrium}

For every bounded utility profile, the constructed mechanism has a pure BNE.
Define
\begin{equation*}
D^T(s)=\mathbf{1}\{s_i\ne\chi_i(s_{-i})\text{ for some }i\}.
\end{equation*}%
At every type $(s_i,E_i)$, let agent $i$ report
\begin{equation*}
r_i=t_i^b=s_i,\qquad z_i^r=\mu_i^r(s_i)\quad(r=1,\ldots,K),\qquad
F_i=E_i,\qquad \ell_i^b=0,
\end{equation*}%
and set
\begin{align*}
c_i^1(s_i)&=\Pr(D^T(s)=1\mid s_i),\\
c_i^\ell(s_i)&=\mathbb{E}\left[\max_{j\ne i}c_j^{\ell-1}(s_j)\mid s_i\right],
\qquad \ell=2,\ldots,L.
\end{align*}%
Evidence ancillarity makes these reports independent of $E_i$, so this is a
pure strategy profile.

To verify optimality, fix a type and any deviation. Replacing a false $r_i$
by $s_i$ is strictly profitable by the comparison following \textup{(\ref%
{RHOSCALE})}. With $r_i=s_i$, deactivating an active challenge is strictly
profitable: the centered part of \textup{(\ref{CHALLTRANS})} has expectation
zero, and $C$ dominates every change in the other unscaled auxiliary terms.
Every feasible evidence report is compatible with $s_i$, so replacing it by
$E_i$ leaves $D$ unchanged and weakly raises the evidence reward. Given the
opponents' maximal evidence and truthful hierarchy reports, strict propriety
and the choice of the $\lambda_r$ make $z_i^r=\mu_i^r(s_i)$ optimal at every
order. Finally, the quadratic reflection scores uniquely elicit the
conditional means displayed above. Applying these comparisons successively
shows that no deviation is profitable. Under nonexclusive information,
$D^T=0$ and every reflection report in this equilibrium is zero. Under
information smallness, these reports can be positive when the truthful
profile activates a prediction-mismatch alarm.

\subsubsection*{Propagation of alarms and truthful hierarchy reports}

\begin{lemma}
\label{lem:USTreflection} Suppose the event $\{D=1\}$ has positive
probability in a belief-closed component. Then, for every type $(s_i,E_i)$
in that component, the equilibrium strategy assigns probability one to
messages that submit maximal evidence, report the true evidence hierarchy 
\begin{equation*}
z_i^r=\mu_i^r(s_i),\qquad r=1,\ldots,K,
\end{equation*}
and report the true state type $r_i=s_i$. The conclusion holds under both
nonexclusive information and information smallness with $\varepsilon_I<1/2$.
\end{lemma}

\begin{proof}
Fix a type $(s_i,E_i)$. Write a message as $(d_i,c_i)$, where $d_i$ collects
the finite report components and $c_i\in[0,1]^L$ is the reflection vector.
Disintegrate $\sigma_i(s_i,E_i)$ with respect to its finite marginal on $d_i$%
. For every $d_i$ assigned positive probability, agent $i$'s expected payoff
as a function of $c_i$ varies only through the mean scores in $\tau_i^R$.
For each $\ell$, 
\begin{equation*}
\mathbb{E}[S^R(c_i^\ell,R_i^\ell)\mid s_i,E_i,d_i] =\mathbb{E}[R_i^\ell\mid
s_i,E_i,d_i]^2 -\left(c_i^\ell-\mathbb{E}[R_i^\ell\mid s_i,E_i,d_i]\right)^2.
\end{equation*}
The unique maximizer is therefore 
\begin{equation}
c_i^\ell=\mathbb{E}[R_i^\ell\mid s_i,E_i,d_i].  \label{CONDREF}
\end{equation}%
The conditional distribution of $c_i$ is concentrated on this vector.

A positive-probability alarm gives the types in that event a strictly
positive conditional mean for $R^{1}$. Equation (\ref{CONDREF}) gives them a
positive first reflection report with positive probability. Whenever a type
assigns positive probability to an opponent using a positive order-$\ell $
reflection report, its conditional mean of $R_i^{\ell +1}$ is positive and
it consequently uses a positive order-$(\ell +1)$ report. Iterating this
implication through the positive-probability interim beliefs defining the
belief closure propagates a positive reflection value to every type in the
closure within $L$ levels.

Every type in the belief closure consequently assigns positive probability
to an opponent's positive reflection value. Its expected $\omega_i$ is
strictly positive. When maximal evidence leaves the refutation status
unchanged, increasing $F_i$ raises $\tau_i^E$. When maximal evidence
activates a refutation, it also weakly raises the first reflection score, so
the same deviation remains profitable. Hence $F_i=E_i$ with probability one
for every type in the component.

With maximal opposing evidence, the distribution of $Y_i^1$ conditional on $%
s_i$ is exactly $\mu_i^1(s_i)$. Strict propriety and the choice of $%
\lambda_1 $ give $z_i^1=\mu_i^1(s_i)$. Suppose inductively that all reports
through order $r-1$ are truthful throughout the component. Then $%
Y_i^r=(E_{-i},\beta_{-i}^{r-1}(s_{-i}))$ has distribution $\mu_i^r(s_i)$, so
strict propriety and the choice of $\lambda_r$ give $z_i^r=\mu_i^r(s_i)$.
This proves truth through order $K$ with probability one.

Truthful hierarchy reports give $d_j(z_j)=s_j$ for every agent $j$. Thus the
state type predicted for an agent whose state type is $s_i$ is $%
\chi_i(s_{-i})$. Under nonexclusive information it equals $s_i$ surely.
Under information smallness, \textup{(\ref{ISCOND})} and ancillarity give,
for every $E_i\in\mathbb{E}_i$ with $q_i^e(E_i\mid s_i)>0$, 
\begin{equation*}
\Pr(\widehat s_i=s_i\mid s_i,E_i)\ge1-\varepsilon_I, \qquad \Pr(\widehat
s_i=t_i\mid s_i,E_i)\le\varepsilon_I \quad(t_i\ne s_i).
\end{equation*}
Agent $i$ chooses $r_i$ without observing $s_{-i}$, so reporting $s_i$
reduces the cross-check probability relative to any fixed false report by at
least $1-2\varepsilon_I$. By \textup{(\ref{RHOSCALE})}, the resulting gain
dominates every outcome and auxiliary-transfer change. Hence $r_i=s_i$ with
probability one.
\end{proof}

\subsubsection*{Eliminating all wrong outcomes}

Fix an arbitrary BNE. The next claims apply at every positive-prior state
and throughout its belief closure.

\begin{claim}
\label{cl:USTcons} At every $s\in\supp\pi$, 
\begin{equation*}
\mu_\sigma^s\bigl(\{m:r(m)\in\supp\pi\}\bigr)=1.
\end{equation*}
\end{claim}

\begin{proof}
If the complementary event had positive probability, then the alarm event
would have positive probability. Lemma~\ref{lem:USTreflection} would make
the state-type first reports truthful with probability one throughout the
component, giving $r=s\in\supp\pi$ almost surely. This contradicts the
positive probability of the complementary event.
\end{proof}

\begin{claim}
\label{cl:USTpure} Equilibrium state-type first reports define a pure
state-coordinate map $\alpha(\cdot)=(\alpha_i(\cdot))_i$. For every agent $i$
and state type $s_i$, there is a single $\alpha_i(s_i)\in S_i$ such that,
for every $E_i\in\mathbb{E}_i$ with $q_i^e(E_i\mid s_i)>0$, 
\begin{equation*}
\sigma_i(s_i,E_i)\bigl(\{m_i:r_i=\alpha_i(s_i)\}\bigr)=1.
\end{equation*}
\end{claim}

\begin{proof}
Fix a belief component. If the alarm event has positive probability, Lemma~%
\ref{lem:USTreflection} makes every first report in the component truthful
almost surely. If instead $D=0$ almost surely in the component, each first
report equals the state type predicted by opponents' hierarchy reports
almost surely. Conditional on $s_i$, the distribution of the predicted state
type is independent of $E_i $ by ancillarity and independent of agent $i$'s
mixed report. Equality with probability one makes both the predicted state
type and the first report degenerate at the same value $t_i$. The
cross-check coefficient makes this value the unique first report for every $%
E_i\in\mathbb{E}_i$ satisfying $q_i^e(E_i\mid s_i)>0$. In each case the
report depends only on the state type; denote it by $\alpha_i(s_i)$.
\end{proof}

For every $(s_i,E_i)\in S_i\times\mathbb{E}_i$ with $q_i^e(E_i\mid s_i)>0$,
define $\alpha_i(s_i,E_i)$ by taking $\alpha_i(s_i)$ as its first coordinate
and the equilibrium marginal distribution of $F_i$ conditional on $(s_i,E_i)$
as its second coordinate. At the remaining pairs, choose any feasible
evidence-report distribution and retain the first coordinate $\alpha_i(s_i)$%
. This defines a feasible deception $\alpha=(\alpha_i)_{i\in\mathcal{I}}$.

\begin{claim}
\label{cl:USTnoref} For every $\bar s\in\bar S(s)$, the state-type profile $%
\alpha(\bar s)$ is nonrefutable at $\bar s$.
\end{claim}

\begin{proof}
Suppose some $E_i\in\mathbb{E}_i$ with $q_i^e(E_i\mid\bar s_i)>0$ satisfies $%
E_i\downarrow \alpha(\bar s)$. Type $(\bar s_i,E_i)$ can replace its
evidence component by $E_i$ while leaving the other message components
unchanged. The deviation activates a refutation and sets $\omega_i=1$. The
choice of $a_E$ makes this deviation profitable whenever the equilibrium
evidence report differs from $E_i$ on a positive-probability event. If the
type submits $E_i$ almost surely, the refutation event already has positive
probability. Thus the alarm event has positive probability in the component.

Lemma~\ref{lem:USTreflection} then makes all first reports at $\bar s$
truthful almost surely, so $\alpha(\bar s)=\bar s$. The endowment $E_i$
itself witnesses compatibility with state type $\bar s_i$, which gives a
contradiction.
\end{proof}

\begin{claim}
\label{cl:USTwrong} At every $s\in\supp\pi$, 
\begin{equation*}
\mu_\sigma^s\bigl(\{m:f(r(m))=f(s)\}\bigr)=1.
\end{equation*}
\end{claim}

\begin{proof}
By Claim~\ref{cl:USTpure}, state-type first reports are $\alpha(s)$. Claims %
\ref{cl:USTcons} and \ref{cl:USTnoref} give $\alpha(\bar s)\in \supp\pi$ and
nonrefutability at $\bar s$ throughout $\bar S(s)$. Suppose $f(\alpha(s))\ne
f(s)$.

If an active challenge occurs with positive probability in this component,
it is an alarm, and Lemma~\ref{lem:USTreflection} gives $\alpha(s)=s$, a
contradiction. The maintained wrong-outcome supposition therefore assigns
probability one to inactive challenges.

GNPD implies that the equilibrium-induced deception $\alpha$ is not perfect
on $\bar S(s)$. Lemma~\ref{lem:USTmenu}, applied with $\varphi(t)=\alpha(t)$%
, therefore supplies a state type $t_i$ and a menu vector satisfying \textup{%
(\ref{MENUMARGIN})}. For any $E_i\in\mathbb{E}_i$ satisfying $q_i^e(E_i\mid
t_i)>0$, type $(t_i,E_i)$ can modify its equilibrium message by activating
that challenge and retaining every other coordinate. Its opponents' law is $%
\Lambda_i^\alpha(\cdot\mid t_i)$, so the expected unscaled challenge gain,
after the centering and cost in \textup{(\ref{CHALLTRANS})}, exceeds $3C$.
The change in all other unscaled auxiliary terms is less than $C$ by the
parameter choice, and the outcome and cross-check are unchanged. After
multiplication by the positive scale $\eta_I $, the deviation remains
strictly profitable. This contradicts equilibrium.
\end{proof}

Claim~\ref{cl:USTwrong} proves exact outcome implementation for large $I$
under information smallness and for every $I$ under nonexclusive information.

\subsubsection*{Equilibrium transfers}

The refutation event has probability zero by Claim~\ref{cl:USTnoref}. The
active-challenge event also has probability zero. If it had positive
probability, it would create an alarm, and Lemma~\ref{lem:USTreflection}
would make the state-type first reports, maximal evidence, and the entire
hierarchy truthful almost surely throughout the component. Conditional on
the challenger's state type, the opponents' pseudo types $(r_{-i},F_{-i})$
would then have law $\Lambda_i^*(\cdot\mid r_i)$. The centered part of 
\textup{(\ref{CHALLTRANS})} has expectation zero and the challenge cost is $%
C $. Deactivating the challenge changes the other unscaled auxiliary terms
by less than $C$, which gives a profitable deviation and a contradiction.

First suppose information is nonexclusive. A positive-probability alarm,
after the refutation and challenge events have probability zero, must arise
because $r\notin\supp\pi$, because a hierarchy is inconsistent with a
state-type report, or because a state-type report differs from the state
type predicted by opponents' hierarchy reports. Lemma~\ref{lem:USTreflection}
then makes the hierarchies and state-type first reports truthful almost
surely. Truthful hierarchy reports identify the opponents' state types, and
nonexclusive information makes the state type predicted for each agent
correct. Each source of the alarm consequently has probability zero, a
contradiction. Hence $D=0$ almost surely.

Equation (\ref{CONDREF}) now gives $c_i^1=0$ almost surely. Induction on $%
\ell$ gives $c_i^1=\cdots=c_i^L=0$ almost surely. Thus $\tau_i^E=\tau_i^R=0$
almost surely. The event $D=0$ also makes each hierarchy tuple canonical for 
$r_i$ and gives $r_i=\widehat s_i(z_{-i})$ almost surely. Every bracket in 
\textup{(\ref{HIERTRANSFER})} cancels, so $\tau_i^H=0$ almost surely. The
cross-check and challenge transfers are zero almost surely. Consequently $%
\tau_i=0$ for $\mu_\sigma^s$-almost every message at every $s\in\supp\pi$.
Together with Claim~\ref{cl:USTwrong}, this proves sufficiency in part 2 of
Theorem~\ref{THM_GTS}; necessity was proved above.

Now suppose agents are informationally small, and take $\delta_I\to0$ in 
\textup{(\ref{ETASCALE})}. The absolute value of the entire scaled auxiliary
block in \textup{(\ref{TOTALUST})} is at most $\delta_I$ on every message
profile. It remains to bound the cross-check. On a belief component where $%
D=0$, the cross-check is identically zero. If a component has any
positive-probability alarm, Lemma~\ref{lem:USTreflection} makes the
hierarchy and state-type first report truthful throughout that component.
The state type predicted for agent $i$ is then $\chi_i(s_{-i})$, and \textup{%
(\ref{ISCOND})} bounds its error probability by $\varepsilon_I$. Therefore,
uniformly over equilibria and agents, 
\begin{equation*}
\mathbb{E}_\sigma[|\tau_i|] \le \delta_I+\rho_I\varepsilon_I
\longrightarrow0.
\end{equation*}
Claim~\ref{cl:USTwrong} gives the exact outcome for sufficiently large $I$,
and the displayed bound gives vanishing expected transfers. This proves part
1 of Theorem~\ref{THM_GTS}.

\end{document}